\documentclass{irmaart}
\input{ha.sty}

\usepackage[displaymath]{lineno}
		 \nolinenumbers 

\usepackage[]{xr-hyper}
\usepackage[hypertex,hyperindex,pagebackref,final]{hyperref}

\hypersetup{
		colorlinks = true,
		linkcolor = red,
		anchorcolor = black,
		citecolor = blue, 
		filecolor = cyan,
		menucolor = red,
		runcolor = cyan,
		urlcolor = magenta,
		linkbordercolor = {white},
		linktocpage = true,
		driverfallback=dvipdfm
}

\usepackage{calc}
\newcounter{hours}\newcounter{minutes}
\newcommand\printtime{\setcounter{hours}{\time/60}%
	\setcounter{minutes}{\time-\value{hours}*60}%
	\thehours\,h\,\theminutes}
\newcommand\dateandtime{\today\quad\printtime}

\usepackage{amssymb}
\usepackage{amsmath}
\usepackage{textcomp}
\usepackage{array}
\usepackage{multirow}
\usepackage[latin5]{inputenc}
\usepackage{color}
\usepackage{enumitem}
\usepackage{graphicx}

\newcommand{\bra}[1]{\langle #1|}
\newcommand{\ket}[1]{|#1\rangle}

\newcommand{\cent}[0]{\mbox{\textcent}}
\newcommand{\dollar}[0]{\$}
\newcommand{\comment}[1]{}
\newcommand{\M}[0]{{\cal M}}

\begin{document}

\title{Automata and Quantum Computing}

\author{Andris Ambainis
~~and Abuzer Yakary{\i}lmaz\thanks{A. Yakary{\i}lmaz worked on the chapter when he was in Bo\u{g}azi\c{c}i University in Turkey, University of Latvia, and LNCC in Brazil.}}

\markboth{A. Ambainis, A. Yakary{\i}lmaz}{Automata and Quantum Computing}

\address{University of Latvia, Faculty of Computing, Raina bulv. 19, R\={\i}ga, LV-1586, Latvia \\
  emails:\,\url{ambainis@lu.lv} and  \,\url{abuzer@lu.lv}
  \\[4mm]\upshape{\dateandtime}}


\maketitle


\begin{abstract}
Quantum computing is a new model of computation, based on quantum physics. 
Quantum computers can be exponentially faster than
conventional computers for problems such as factoring.
Besides full-scale quantum computers, more restricted 
models such as quantum versions of finite automata have been studied. 
In this paper, we survey various models of quantum finite automata and their properties.  
We also provide some open questions and new directions for researchers.
\end{abstract}



\localtableofcontents

%
%

\section{Introduction}

Quantum computing\index{quantum computation} combines quantum physics and computer science, by studying computational
models based on quantum physics (which is substantially different from conventional physics)
and building quantum devices which implement those models. 
If a quantum computer is built, it will be able to solve
certain computational problems much faster than conventional computers.

The best known examples of such problems are
factoring and discrete logarithm. These two number theoretic problems
are thought to be very difficult for conventional computers but can be solved efficiently
(in polynomial time) on a quantum computer \cite{Sh94}. 
Since several widely used cryptosystems 
(such as RSA and Diffie-Hellman) are based on the difficulty of factoring or discrete logarithm,
a quantum computer would be able to break those cryptosystems, shaking up the foundations of cryptography.

Another, equally surprising discovery was made in 1996, by Lov Grover \cite{Gr96} who designed a quantum
algorithm that solves a general exhaustive search problem with $ N $ 
possible solutions in time $ O(\sqrt{N}) $.
This provides a quadratic speedup for a range of search problems, from problems
that are solvable in polynomial time classically to NP-complete problems.

Many other quantum algorithms have been discovered since then. (More information can be found in surveys \cite{BD10,Mo10} and the ``Quantum Algorithm Zoo" website \cite{Jo15}.)

Given that finite automata are one of the most basic models of computation, 
it is natural to study them in the quantum setting. Soon after the discovery of Shor's factoring algorithm \cite{Sh94}, the first models of quantum finite automata 
 (QFAs) appeared \cite{KW97,MC00}. A number of different models
and questions about the power of QFAs and their properties have been studied since then.

In this chapter, we cover most of this work. We particularly focus on the results which show advantages of QFAs over their classical\footnote{In the context of quantum computing, ``classical" means ``non-quantum". For finite automata, this usually means deterministic or probabilistic automaton.} counterparts because those results show how ``quantumness" adds power to the computational models.  

We note that some of early research on QFAs also claimed that, in some contexts, QFAs can be weaker than their classical counterparts. This was due to the first definitions of QFAs being too restricted \cite{Wat09B}. Quantum computation is a generalization of classical computation \cite{Wat09A} and QFAs should be able to simulate classical finite automata, if we define QFAs in sufficiently general way. Therefore, we particularly emphasize the most general model of QFAs that fully reflect the power of quantum computation.

We begin with an introductory section (Sec. \ref{sec:quantum-computation}) on basics of quantum computation for  readers who are not familiar with it. Then, we give the basic notation and conventions used throughout this chapter 
(Sec. \ref{sec:preliminaries}). After that, Sec. \ref{sec:1QFAs} presents the main results on 1-way QFAs and Sec. \ref{sec:2QFAs} presents the main results on 2-way QFAs. Each of those sections covers different models of QFAs that have been proposed, the classes of languages that they recognize, 
the state complexity in comparison to the corresponding classical models, and decidability and undecidability results. In Sec. \ref{sec:other-models}, we describe the results about QFAs in less conventional models or settings (for example, interactive proof systems with a QFA verifier or QFAs  
augmented with extra resources beyond the usual quantum model).  We conclude with a discussion of directions for future research in Sec. \ref{sec:conclusion}.

We also refer the reader to \cite{SayY14A} for an introductory paper on quantum automata and to \cite{QLMG12} for another survey on quantum automata, and to 
\begin{center}
http://publication.wikia.com/wiki/Quantum\_automata
\end{center}
for a list of published papers on quantum automata.

\section{Mathematical background} \label{sec:quantum-computation}

In this section, we review the basics of quantum computation.\index{quantum computation}
We refer the reader to \cite{NC00} for more information.

{\bf Quantum systems.}\index{quantum system}
The simplest way towards understanding the quantum model is by thinking of  as a generalization of probabilistic systems.
If we have a probabilistic system with $m$ possible states $1, 2, \ldots, m$, we can describe it by a probability distribution $p_1, \ldots, p_m$ over those $m$ possibilities. The probabilities $p_i$ must be nonnegative real numbers and satisfy $p_1+\cdots+p_m=1$. In the quantum case,
the probabilities $p_1, \ldots, p_m$ are replaced by amplitudes $\alpha_1, \ldots, \alpha_m$. The amplitudes can be complex numbers and must satisfy $|\alpha_1|^2+\cdots+|\alpha_m|^2=1$. 

More formally, let us consider a quantum systems with 
$m$ basis states
(for some finite $m$) which we denote by $\ket{q_1}, \ket{q_2}, \ldots, \ket{q_m}$. 
A state of such a system is a linear combination of basis states with complex coefficients (called {\em amplitudes})\index{amplitude}
\begin{equation}
\label{eq:superpose} 
\ket{\psi}=\alpha_1\ket{q_1}+\alpha_2\ket{q_2}+\cdots+\alpha_m\ket{q_m}
\end{equation}
that must satisfy $|\alpha_1|^2+\cdots+|\alpha_m|^2=1$. 
We say that $\ket{\psi}$ is
a {\em superposition}\index{superposition} of $\ket{q_1}$, $\ldots$, $\ket{q_m}$.
For example, if we have a system with 2 basis states $\ket{0}$ and $\ket{1}$, 
some of the possible superpositions are 
$\frac{4}{5} \ket{0}+\frac{3}{5} \ket{1}$, $\frac{4}{5} \ket{0}-\frac{3}{5} \ket{1}$, and $\frac{1}{\sqrt{2}} \ket{0}+\frac{1}{\sqrt{2}} \ket{1}$.

We can view a \index{quantum state} $\ket{\psi}$ as a vector consisting of amplitudes:
\[ \ket{\psi}= \left( \begin{array}{c} \alpha_1 \\ \alpha_2 \\ \vdots \\ \alpha_m 
\end{array} \right) .\]
Then, the basis states $\ket{q_i}$ are vectors with 1 in the $i^{\rm th}$ component and
0 everywhere else and (\ref{eq:superpose}) can be interpreted as the addition of vectors. 
The length of the vector $\ket{\psi}$ is $\|\psi\|= \sqrt{|\alpha_1|^2+\cdots+|\alpha_m|^2}$. Thus, in vector language, the requirement that 
$|\alpha_1|^2+\cdots+|\alpha_m|^2=1$ is equivalent to just saying that $\|\psi\|=1$.
That is, a quantum state is a vector of length 1.

{\bf Unitary transformations.}\index{unitary transformation}
A transformation on a quantum state is specified by a transformation matrix $U$. If the state before the transformation is $\ket{\psi}$, the state after the transformation is $U\ket{\psi}$. A transformation is valid (allowed by the rules of quantum physics) if and only if $\|\psi\|=1$ implies $\|U\ket{\psi}\|=1$. Transformation matrices which satisfy this constraint are called {\em unitary}.

A transformation $U$ can be also specified by describing $U\ket{q_1}, \ldots, U\ket{q_m}$. Then, for any $\ket{\psi}=\alpha_1\ket{q_1}+\alpha_2\ket{q_2}+\cdots+\alpha_m\ket{q_m}$, we have 
\[ U\ket{\psi}=\alpha_1 U\ket{q_1}+ \alpha_2 U\ket{q_2} + \cdots + \alpha_m U\ket{q_m} .\] 
For example, if we have a system with 2 basis states $\ket{0}$ and $\ket{1}$, 
we can specify a transformation $H$ by saying that $H$ maps
\begin{equation} 
\label{eq:had}
\ket{0} \rightarrow \frac{1}{\sqrt{2}} \ket{0}+\frac{1}{\sqrt{2}} \ket{1} \mbox{~~and~~} \ket{1} \rightarrow \frac{1}{\sqrt{2}} \ket{0}-\frac{1}{\sqrt{2}} \ket{1} .
\end{equation}
This determines how $H$ acts on superpositions of $\ket{0}$ and $\ket{1}$, 
For example, (\ref{eq:had}) implies that $H$ maps $\frac{4}{5}\ket{0}-\frac{3}{5}\ket{1}$ to
\[ \frac{4}{5} \left( \frac{1}{\sqrt{2}} \ket{0}+\frac{1}{\sqrt{2}} \ket{1} \right) - \frac{3}{5} \left( \frac{1}{\sqrt{2}} \ket{0}-\frac{1}{\sqrt{2}} \ket{1} \right) 
= \frac{7}{5\sqrt{2}} \ket{0}+\frac{1}{5\sqrt{2}} \ket{1} .\]

{\bf Measurements.}\index{quantum measurement}
To obtain information about a quantum state, we have to measure it.
The simplest measurement is observing 
$ \ket{ \psi} =\alpha_1\ket{q_1}+\alpha_2\ket{q_2}+\cdots+\alpha_m\ket{q_m}$
with respect to $\ket{q_1}, \ldots, \ket{q_m}$.
It gives $\ket{q_j}$ with probability $ | \alpha_j | ^2$.
($\|\psi\|=1$ guarantees that probabilities of different outcomes sum to 1.)
After the measurement, the state of the system changes to $\ket{q_j}$
and repeating the measurement gives the same state $\ket{q_j}$.
For example, observing $\frac{4}{5} \ket{0}+\frac{3}{5} \ket{1}$
gives 0 with probability $(\frac{4}{5})^2 =\frac{16}{25}$ and 1 with probability $(\frac{3}{5})^2 =\frac{9}{25}$.

{\bf Partial measurement.} 
In the context of QFAs, it may be the case that we only need to know whether the state $q_i$ is an accepting state or not. In this case, we can perform a partial measurement. 
Let $Q_1, \ldots, Q_k$ be a partition of $\{q_1, q_2, \ldots, q_m\}$ into disjoint subsets. Then, measuring a state $ \ket{ \psi} =\alpha_1\ket{q_1}+\alpha_2\ket{q_2}+\cdots+\alpha_m\ket{q_m}$ with respect to this partition gives result $Q_i$ with probability
$p_i = \sum_{q_j\in Q_i} |\alpha_j|^2$ and the state after the measurement is 
\[ \sum_{q_j\in Q_i} \frac{\alpha_j}{\sqrt{p_i}} \ket{q_j} .\]
For example, if the quantum state is $\frac{1}{2}\ket{1}+\frac{1}{2}\ket{2}+\frac{1}{2}\ket{3}+\frac{1}{2}\ket{4}$ and the partition is $Q_1=\{1, 2\} $ and $ Q_2=\{3, 4\}$, a partial measurement would give the result $Q_1$ with probability 
$(\frac{1}{2})^2+(\frac{1}{2})^2 = \frac{1}{2}$ and that state after the measurement is
\[ \frac{1}{\sqrt{2}} \ket{1} + \frac{1}{\sqrt{2}} \ket{2} .\]
Such a measurement tells whether a QFA accepts a string and, at the same time, preserves
the part of the quantum state which consists of accepting states (or the part of the quantum state which consists of nonaccepting states).

{\bf Dirac notation.}\index{Dirac notation}
As already mentioned above, we can view a quantum state $\ket{\psi}$ as a vector consisting of amplitudes.
$\bra{\psi}$ denotes the conjugate transpose of this vector:
$ \bra{\psi}= \left( \begin{array}{cccc} \alpha^*_1 & \alpha^*_2 & \cdots & \alpha^*_m
\end{array} \right) ,$
where $\alpha^*_i$ denotes the conjugate transpose of the complex number $\alpha_i$.
(If $\alpha_i$ is real, then $\alpha^*_i=\alpha_i$.)
If we multiply $ \ket{\psi} $ with $\bra{\psi}$ (according to the usual rules
for matrix multiplication), we get a $m\times m$ matrix:
\[ \ket{\psi} \bra{\psi} = \left( \begin{array}{cccc} 
\alpha_1 \alpha^*_1 & \alpha_1 \alpha^*_2 & \cdots & \alpha_1 \alpha^*_m \\
\alpha_2 \alpha^*_1 & \alpha_2 \alpha^*_2 & \cdots & \alpha_2 \alpha^*_m \\
\vdots & \vdots & \ddots & \vdots \\
\alpha_m \alpha^*_1 & \alpha_m \alpha^*_2 & \cdots & \alpha_m \alpha^*_m \\
\end{array} \right) .\] 
This is {\em the density matrix}\index{density matrix}\index{matrix!density} of the state $\psi$. 

{\bf Mixed states.} 
A {\em mixed state (or mixture)} $(p_j, \ket{\psi_j})$ 
is a probabilistic combination of several quantum states $\ket{\psi_j}$,
with probabilities $p_j$ (where $p_{j} \geq 0$ for all $j$ and $\sum_j p_j =1$).
For such a state, its density matrix is just the sum of density matrices 
of $\ket{\psi_j}$, weighted by their respective probabilities:
\[ \rho = \sum_j p_j \ket{\psi_j} \bra{\psi_j} .\]
If we measure a mixed state $(p_j, \ket{\psi_j})$,
the probabilities of different measurement outcomes can be calculated from
the density matrix $\rho$. Thus, $\rho$ provides a complete description of a mixed state:
there may be multiple decompositions $(p_j, \ket{\psi_j})$ that give the same matrix $\rho$ but
they are all equivalent with respect to any measurement that we may perform.

{\bf Superoperators.}\index{quantum superoperator}
In general, we can perform a sequence of unitary transformations and measurements
on a quantum state, with each transformation possibly depending on the results of the previous 
measurements. Such sequence is called a {\em completely positive superoperator} or CPSO. 

Alternatively, a completely positive superoperator can be described by 
a sequence of $ m \times m $ matrices $ \mathcal{E}=\{E_1, \ldots, E_k \}$ (called {\em Kraus operators}) 
such that $\sum_{i=1}^k E_i^{\dagger} E_i  = I$. Such a CPSO maps a mixed state with the density matrix
$\rho$ to a mixed state with the density matrix 
$\rho'= \mathcal{E}(\rho) = \sum_{i=1}^k E_i \rho E_i^{\dagger}$.

The two definitions are equivalent: for any sequence of unitary transformations and measurements,
there is a set of Kraus operators 
$E_1, \ldots, E_k$ which produces the same result and the other way around.
A bistochastic quantum operation, 
say $ \mathcal{E} = \{E_1,\ldots,E_m\} $, is a special kind of superoperator satisfying both  $\sum_{i=1}^k E^\dagger_i E_i  = I$ and  $\sum_{i=1}^k E_i E^{\dagger}_i  = I$. 

\section{Preliminaries} 
\label{sec:preliminaries}

\noindent \textbf{Basic notations:} Throughout the chapter,
\begin{itemize}
	\item $ \Sigma $ is the input alphabet not containing the end markers $ \cent $ and $ \dollar $
		and $ \tilde{\Sigma} = \Sigma \cup \{ \cent, \dollar \} $.
		\item For a given string $ w $, $ |w| $ is the length of $ w $, 
		$ w_{i} $ is the $ i^{th} $ symbol of $ w $, and $ \tilde{w} $ represents the string $ \cent w \dollar $.
		\item $ Q $ is the set of internal states,	where $ q_{1} $ is the initial state.
		$ Q_{a} \subseteq Q $ is the set of accepting states.
	 \item $ f_{\mathcal{M}}(w) $ is the accepting probability 
(or the accepting value) 
		of $ \mathcal{M} $ on the string $ w $.
	\item For a given vector (row or column) $ v $, $ v[i] $ is the $ i^{th} $ entry of $ v $.	
	\item For a given matrix $ A $, $ A[i,j] $ is the $ (i,j)^{th} $ entry of $ A $.	
\end{itemize}

~\\
\noindent 
\textbf{Language recognition:}
Let $ \mathcal{M} $ be a machine and $ \lambda \in \mathbb{R} $.
The language $ L \subseteq \Sigma^{*} $ recognized by $ \mathcal{M} $ with \textit{(strict) cutpoint} 
or \textit{nonstrict cutpoint} 
$ \lambda $ is defined as
\begin{equation*}
      L = \{ w \in \Sigma^{*} \mid f_{\mathcal{M}}(w) > \lambda \} \mbox{ or } L = \{ w \in \Sigma^{*} \mid f_{\mathcal{M}}(w) \geq \lambda \} \mbox{, respectively.}
\end{equation*}
The language $ L \subseteq \Sigma^{*} $ is said to be recognized by $ \mathcal{M} $  with\textit{unbounded error}\index{unbounded error}\index{error!unbounded} if there exists a cutpoint  $ \lambda $ such that $ L $ is recognized by $ \mathcal{M} $ with strict or nonstrict cutpoint $ \lambda $.

In the followings, we assume that $ f_{\mathcal{M}}(w) \in [0,1] $  $ \forall w \in \Sigma^{*} $.
The language $ L \subseteq \Sigma^{*} $ recognized by $ \mathcal{M} $ with \textit{positive or negative  one-sided unbounded error} 
is defined as
\begin{equation*}
      L = \{ w \in \Sigma^{*} \mid f_{\mathcal{M}}(w) > 0 \} \mbox{ or }  L = \{ w \in \Sigma^{*} \mid f_{\mathcal{M}}(w) = 1 \}  \mbox{, respectively.}
\end{equation*}
The language $ L \subseteq \Sigma^{*} $ is said to be recognized by $ \mathcal{M} $ with \textit{error bound} $ \epsilon $ ($ 0 \le \epsilon < \frac{1}{2} $) if (i) $ f_{\mathcal{M}}(w) \ge 1 - \epsilon $ when $ w \in L $ and (ii) $ f_{\mathcal{M}}(w) \leq  \epsilon $ when $ w \notin L $. This notion is also known as recognition with \textit{bounded error}.\index{bounded error}\index{error!bounded} Moreover, in the case of \textit{positive one-sided bounded error}, 
$ f_{\mathcal{M}}(w) = 0 $ when $ w \notin L $; and, in the case of \textit{negative one-sided bounded error}, $ f_{\mathcal{M}}(w) = 1 $ when $ w \in L $.

~\\
\noindent
\textbf{Transition amplitudes and probabilities:}
In both probabilistic and quantum finite automata \cite{Ra63,Pa71,KW97}, 
the transition values (probabilities or amplitudes)\index{amplitude}
are traditionally allowed to be in $ \mathbb{R} $ and in $ \mathbb{C} $, respectively.
On the other hand, the transition values of Turing machines (TMs) \cite{FK94,BV97,Wat03} 
are often selected from the restricted subsets of $ \mathbb{R} $ or $ \mathbb{C} $.
For example, for probabilistic Turing machines, it is often assumed that, at each step,
there are two possible choices and the machine chooses each of them with probability 1/2.

In this chapter, we assume the most general model possible.
That is, unless specified otherwise, the transition values of probabilistic (or quantum) machines
are supposed to be in $ \mathbb{R} $ (or $ \mathbb{C} $). 
We use $ \dddot{\mathbb{R}}$, $ \tilde{\mathbb{C}} $, and $ \mathbb{A} $ to denote the computable real numbers, 
efficiently computable complex numbers, 
and algebraic numbers, respectively \cite{BV97,Wat03}. 


~\\ \noindent
\textbf{Classical finite automata:}
We will compare quantum automata with following models of classical automata:
\begin{itemize}
\item
deterministic finite automaton 
(DFA),
\item
nondeterministic finite automaton (NFA), 
\item
probabilistic finite automaton\index{probabilistic finite automaton} \index{automaton!probabilistic} (PFA), and,
\item
generalized finite automaton\index{generalized finite automaton}\index{automaton!generalized} (GFA) of Turakainen \cite{Tur69}.
\end{itemize}
The first three models are quite widely known. Each of them can be studied both in 1-way version (where the head of the automaton moves from the left to the right) 
and in 2-way version (where the automaton is allowed to move in both direction).
We refer to 1-way models as 1DFA, 1NFA and 1PFA and to 2-way models as 2DFA, 2NFA and 2PFA.\index{two-way automaton}\index{automaton!two-way}

A 1-way probabilistic finite automaton\index{probabilistic finite automaton}\index{automaton!probabilistic} (1PFA) \cite{Ra63} can be described by a 5-tuple $ \mathcal{P}=(Q,\Sigma,\{ A_{\sigma} \mid \sigma \in \tilde{\Sigma} \},q_{1},Q_{a}), $ where $ A_{\sigma} $ is the transition matrix,  i.e. $ A_{\sigma}[j,i] $ is the probability of the transition from state $ q_{i} $ to state $ q_{j} $ when reading symbol $ \sigma $.
We require that all $A_{\sigma}$ are stochastic:\index{stochastic matrix} $A_{\sigma}[j,i]\geq 0$ and,
for any $i$, $\sum_j A_{\sigma}[j,i]=1$.

The computation of a 1PFA can be traced by a probability vector $ v $ in which $ v[i] $ is the probability of being in state $ q_{i} $. For a given input string $ w \in \Sigma^{*} $, $ \tilde{w} = \cent w \dollar $ is read symbol by symbol, $  v_{i} = A_{\tilde{w}_{i}} v_{i-1} $, where $ 1 \le i \le | \tilde{w} | $ and $ v_{0} $ is the initial state vector whose first entry is equal to 1. 
The acceptance probability of $ \mathcal{P} $ on string $w$ is defined as 
\begin{equation*}
      f_{\mathcal{P}}(w) = 
              \sum_{q_{i} \in Q_{a}} v_{|\tilde{w}|}[i].
\end{equation*}

If we allow the transition matrices $A_{\sigma}$ to be arbitrary matrices consisting of   arbitrary real numbers, we obtain a generalized finite automaton\index{generalized finite automaton}\index{automaton!generalized} (GFA) \cite{Tur69}. Then, the range of $ f_{\mathcal{G}}(\cdot) $ is real numbers and it is called \textit{accepting value} instead of accepting probability. 

1DFAs, 1NFAs, 2DFAs, 2NFAs and 1PFAs with bounded error all recognize the same class of languages: \textit{regular languages} $ \mathsf{REG} $. 
On the other hand, 2PFAs can recognize some nonregular languages \cite{Fr81} with bounded error, such as $  \mathtt{UPAL}=\{a^{n}b^{n} \mid n>0 \} $, but requires exponential expected runtime \cite{DS90}. 

The languages recognized by 1PFAs with cutpoint (nonstrict cutpoint) 
form the class of \textit{stochastic languages}\index{stochastic language}\index{language!stochastic} (\textit{co-stochastic languages}), denoted as $ \mathsf{S} $ ($ \mathsf{coS} $).  
$ \mathsf{S}  \cup \mathsf{coS} $ forms the class of \textit{unbounded-error stochastic languages} 
($ \mathsf{uS} $). 
1GFAs are equivalent to 1PFAs: the classes of languages recognized by 1GFAs with cutpoint  and nonstrict cutpoint are also $ \mathsf{S} $ and $ \mathsf{coS} $, respectively \cite{Tur69}. This equivalence
makes 1GFAs useful for proving facts about 1PFA.
Moreover, any language $ L $, defined as $  L =  \{ w \mid f_{\mathcal{P}}(w) \neq \frac{1}{2} \} $ for a 1PFA  $ \mathcal{P} $, is called \textit{exclusive stochastic language} 
and $ \mathsf{S ^{\neq} } $ denotes the class of all such languages. The complement of $ \mathsf{S^{\neq}} $ is denoted $ \mathsf{S^{=}} $. The class $ \mathsf{S^{=}_\mathbb{Q}} $ is a subset of $ \mathsf{S^{=}} $ defined by 1PFAs with rational-valued transitions. The class $ \mathsf{S^{\neq}_\mathbb{Q}} $ is defined similarly.

\section{One-way QFAs}
\label{sec:1QFAs}

A number of different definitions of one-way QFAs\index{one-way quantum automaton}\index{quantum automaton!one-way}\index{automaton!one-way quantum} (1QFAs) have been proposed over the years.
They can all be described in a similar framework, by specifying a 5-tuple
\[ 
	(\Sigma,Q, \{T_{\sigma} \mid \sigma \in \hat{\Sigma} \}, q_1, R ),
\]
with the following specifications:
\begin{itemize}
\item
$Q$ is a finite set of (classical) states. The (quantum) state of a 1QFA can be any superposition\index{superposition} of basis states $ \{ \ket{q} \mid q \in Q \} $: $\ket{\psi}=\sum_{q\in Q} \alpha_q \ket{q}$.
In mixed state 
models, the 1QFA can also be in a mixture $(p_i, \ket{\psi_i})$ of such states.
\item
$\ket{\psi_{0}} =\ket{q_1} $ is the initial state of the 1QFA.
\item
For each symbol $ \sigma \in \tilde{\Sigma} $, we have a corresponding transformation $T_{\sigma}$ on 1QFA's current 
state. In simpler models, $T_{\sigma}$ is a unitary transformation, 
denoted $U_{\sigma}$. In more general models,
$T_{\sigma}$ can be a sequence of unitary transformations and measurements,\index{quantum measurement} with the next operations 
in the sequence depending on the previous ones.
\item
$R$ is a rule for determining how the 1QFA accepts the strings. Typically, 
$R$ is specified by a set of accepting states ($Q_{a}\subseteq Q$) and measuring the final 
state of the QFA in the standard basis. If an accepting state $ q\in Q_{a}$ is obtained,
the automaton accepts. Otherwise the automaton rejects.
\end{itemize}

The models differ in the set of transitions $T_\sigma$ and acceptance rules $R$ that are allowed.

{\bf Example.}
Let $p$ be an odd number.
Consider the following 1QFA $\mathcal{M}$ in one symbol alphabet\index{unary language}\index{language!unary} $\Sigma=\{a\}$. The set of states is $Q=\{q_1, q_2\}$ and the initial state is $\ket{\psi_{0}}= \ket{q_1}$. The transformations on the end-markers are identities and the transformation $U_a$ that corresponds to reading $a$ is defined by
\[
	\begin{array}{rcr}
		U_a |q_1\rangle & = & \cos\phi|q_1\rangle + \sin\phi|q_2\rangle
		\\
		U_a |q_2\rangle & = & -\sin\phi|q_1\rangle + \cos\phi|q_2\rangle
	\end{array}, \mbox{ where } \phi=\frac{2\pi}{p}.
\]
The acceptance rule $R$ is as follows: At the end of computation, we measure the state. If the result is $q_1$, we accept. Otherwise, we reject. (In other words, we have $Q_{a}=\{q_1\}$.)

The initial state is $\ket{q_1}$. After reading the first symbol $a$, the quantum state becomes
$ \ket{\psi_1}= \cos\phi|q_1\rangle + \sin\phi|q_2\rangle .$
After reading the second symbol $a$, it becomes 
\begin{eqnarray*}
	U_a\ket{\psi_1} & = & \cos \phi ~ U_a\ket{q_1} + \sin \phi ~ U_a\ket{q_2} 
	\\
	& = & \cos \phi ( \cos\phi \ket{q_1} + \sin\phi \ket{q_2} ) + \sin \phi ( -\sin\phi \ket{q_1} + \cos\phi \ket{q_2} )
	\\
	& = & \cos 2\phi \ket{q_1} + \sin 2 \phi \ket{q_2}.
\end{eqnarray*}
We can show that each next application of $U_a$ also rotates the state in the plane formed by $\ket{q_1}, \ket{q_2}$ by an angle of $\phi=\frac{2\pi}{p}$, as shown in Figure \ref{fig:2state-qfa}. Thus, we have

\begin{figure}[!htb]
\centering
\includegraphics[scale=0.3]{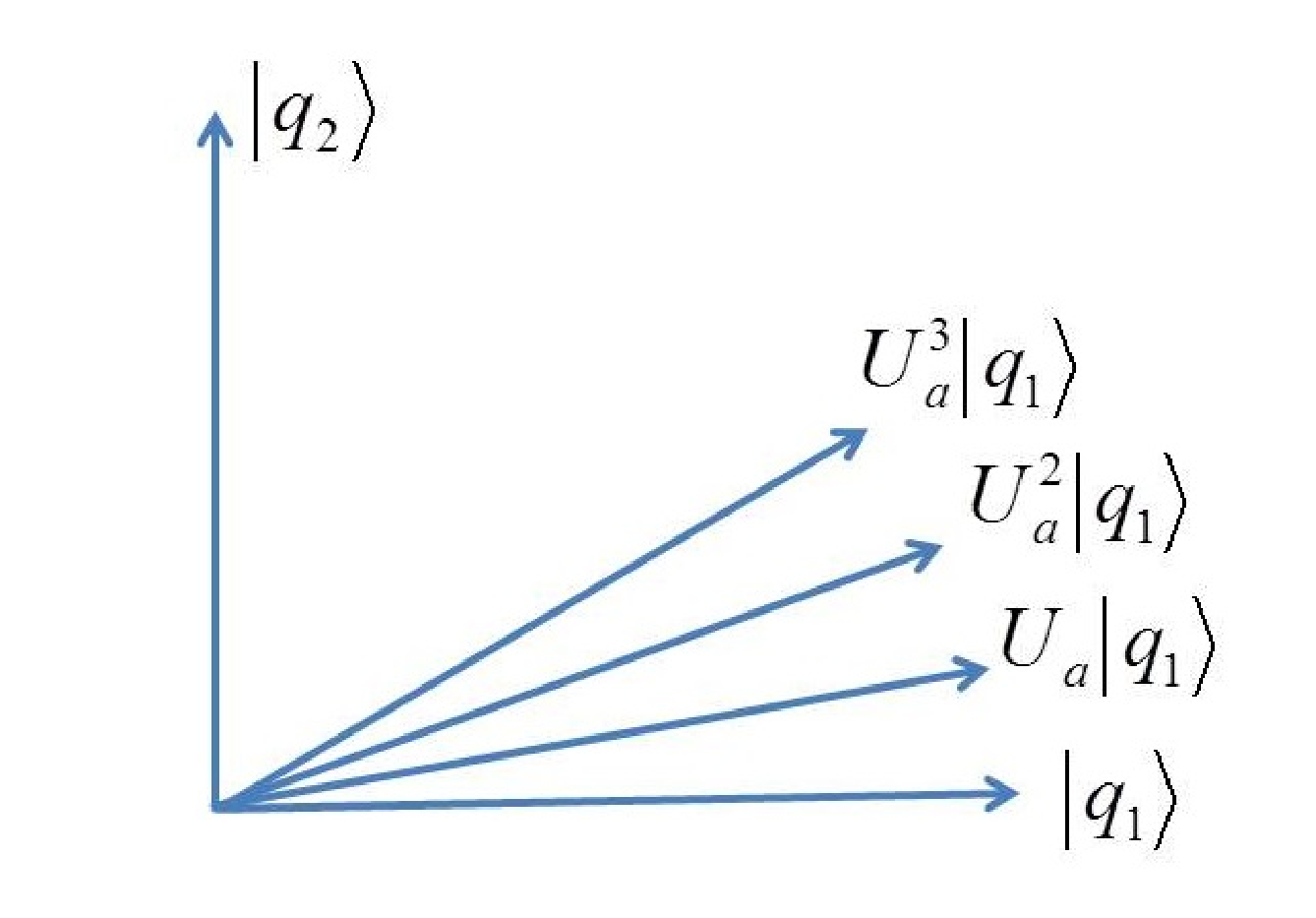}
\caption{2 state 1QFA $\mathcal{M}$}
\label{fig:2state-qfa}
\end{figure}

\begin{lemma} \cite{AN09}
	\label{lem:2-state-MCQFA}
	After reading $a^j$, the state of $\mathcal{M}$ is
	$ \cos\left(\frac{2\pi j}{p}\right) \ket{q_1}
		+\sin\left(\frac{2\pi j}{p}\right) \ket{q_2}.$
\end{lemma}
Therefore, $\mathcal{M}$ accepts $a^j$ with probability $\cos^2 (\frac{2\pi j}{p})$. If $p$ divides $j$, $\cos^2(\frac{2\pi j}{p})=\cos^2 0 = 1$. Otherwise, $\cos^2(\frac{2\pi j}{p}) < 1$. So, for any odd $p>2$, $\M$ recognizes
\begin{equation}
	\label{eq:language:MOD-p}
	\mathtt{MOD_p} = \{a^i | \mbox{$p$ divides $i$}\}
\end{equation}
with negative one-sided error bound $ \cos^2(\frac{\pi}{p}) $.

{\bf Models.} We now describe some important 1QFA models in the order from the most restrictive to the most general. Let $w$ be the input string.
\begin{enumerate}
\item
\textit{Moore-Crutchfield quantum finite automaton}\index{quantum automaton!Moore-Crutchfield}\index{automaton!Moore-Crutchfield quantum} (MCQFA) 
\cite{MC00} is the most restricted of known QFA models.
In this model, the transformations $ T_\sigma $'s have to be unitary\index{unitary transformation} 
($U_{ \sigma }$).
The acceptance rule $R$ is of the form described above.
We measure the state of the QFA after reading $\tilde{w}$ and accept if the measurement\index{quantum measurement} gives $ q\in Q_{a}$.
\item
\textit{Kondacs-Watrous quantum finite automaton}\index{quantum automaton!Kondacs-Watrous}\index{automaton!Kondacs-Watrous quantum} (KWQFA) \cite{KW97}
is a model in which $ T_\sigma $'s still have to be unitary ($U_\sigma$) but the acceptance
rule $R$ involves measurements after every step. 
The state set $Q$ is partitioned into the set of accepting states 
$Q_{a}$, the set of rejecting states $Q_{r}$, and the set of non-halting states $Q_{n}$, i.e. $Q=Q_{a}\cup Q_{r}\cup Q_{n}$.
After reading each symbol $\tilde{w}_i$, we perform a partial measurement 
whether the state is in $Q_{a}$, $Q_{r}$, or $Q_{n}$.
If $Q_{a}$ ($ Q_r $) is obtained, the computation is terminated and the input string is accepted (rejected).
If $Q_{n}$ is obtained, the computation is continued by reading the next symbol $\tilde{w}_{i+1}$ when $ i < | \tilde{w} | $, and the computation is terminated and the input string is rejected when $ i=|\tilde{w}| $.
\item \textit{Latvian quantum finite automaton}\index{quantum automaton!Latvian}\index{automaton!Latvian quantum} (LaQFA) \cite{ABGKMT06} is a model in which each $T_\sigma$ can be a sequence $U_1$, $M_1$,  $\ldots$, $U_m$, $M_m$ consisting of unitary transformations $U_1, \ldots, U_m$ and measurements $M_1, \ldots, M_m$. It is required that the transformations in the sequence are independent of the outcomes of measurements $M_i$ ($ 1 \leq i \leq m $). The acceptance is by measuring the state of the QFA after reading $ \tilde{w} $ and accepting if the measurement gives $q\in Q_{a}$.
\item
\textit{Nayak quantum finite automaton}\index{quantum automaton!Nayak}\index{automaton!Nayak quantum} (NaQFA) \cite{Nay99} combines KWQFA and LaQFA. The transformations $T_\sigma$'s of NaQFA as for LaQFA, the acceptance rule $R$ is as for KWQFA.
\textit{Bistochastic quantum finite automaton}\index{quantum automaton!bistochastic}\index{automaton!bistochastic quantum} (BiQFA) \cite{GKK11} is a generalization of NaQFA such that the transformations $T_\sigma$'s can be bistochastic quantum operations
and the acceptance rule $R$ is as for KWQFA. 
\item
\textit{General one-way quantum finite automaton}\index{quantum automaton!general}\index{automaton!general quantum} (1QFA) \cite{Hir10,YS11A} allows each $T_\sigma$ to be 
a sequence $U_1, M_1$, $\ldots$, $U_m, M_m$ of unitary transformations and measurements. Moreover, each transformation $U_i$ or measurement $M_i$ can depend on the outcomes  of previous measurements $M_{i-1}$. This is the most general model. It has been discovered several times in different but 
equivalent forms: quantum finite automaton with ancilla qubits (QFA-A)\index{quantum automaton!with ancilla qubits}\index{automaton!quantum with ancilla qubits} \cite{Pa00}, fully quantum finite automaton (CiQFA)\index{quantum automaton!fully}\index{automaton!fully quantum} \cite{Cia01}, quantum finite automaton with control language\index{quantum automaton!with control language}\index{automaton!quantum with control language} (QFA-CL) \cite{BMP03A}, and one-way finite automaton with quantum and classical states\index{quantum automaton!with quantum and classical states}\index{automaton!with quantum and classical states} (1QCFA) \cite{ZQLG12}.
\end{enumerate}

Why do we have so many models? Initially, researchers did not recognize the power that comes from performing sequences of unitary transformations and measurements. For this reason, the first models of QFAs were defined in an unnecessarily restrictive form.

Which is the right model? Physically, we can perform any sequence of measurements 
and unitary transformations. 
Hence, 1QFAs should be the main model. 
Among the more restricted models, MCQFA and LaQFA can be motivated by the fact that 
measuring the quantum state during each transformation $T_a$ may be difficult. 
Because of that, it may be interesting to consider models in which measurements are restricted. 
One natural restriction is to allow only one measurement at the end of the computation, as in MCQFA. 
The other possibility is to allow intermediate measurements after each $T_a$, as long as the rest of computation does not depend on their outcomes. 
Such measurements are easier to realize than general measurements (for example, this is the case in {\em liquid state NMR} 
quantum computing \cite{NMR06}). 
This leads to LaQFA.

There is no compelling physical motivation behind KWQFAs and NaQFAs. These models allow to stop the computation depending on the result of an intermediate measurement. 
If we are able to do that, it is natural that we also allow the next transformations to depend on the measurement 
outcome - which leads to the most general model of 1QFAs.

\subsection{Simulations} 
\label{sec:simulations}

In this section, we present some basic simulation results that relate the power of 1QFA to the power of their classical counterparts. First, any probabilistic automaton can be transformed into an equivalent QFA, if the model of QFA is sufficiently general.

\begin{theorem} 
	\label{thm:1PFA-simulated-by-1QFA}
	For a given $ n $-state 1PFA $ \mathcal{P} $, there exists an $ n $-state 1QFA $ \mathcal{M} $ such that $ f_{\mathcal{P}}(w)=f_{\mathcal{M}}(w) $ for any $ w \in \Sigma^{*} $.
\end{theorem}

This result easily follows from the fact that stochastic operators (transformation $A_{\sigma}$ in a 1PFA) are a special case of superoperators (used by 1QFAs).\index{stochastic matrix}\index{quantum superoperator}
The proof of the theorem can be found in \cite{Hir10,YS11A,SayY14A} but it has been
known as a folklore result in quantum computing community long before that.
The same result is also valid in many other settings, for example, for probabilistic and quantum Turing machines \cite{Wat03}.

The second simulation shows how to convert a 1QFA to a GFA with a quadratic increase in the number of internal states.
\begin{theorem}
	\label{thm:1QFA-is-a-GFA}
	\cite{MC00,LQ08,YS11A}
	For a given $ n $-state 1QFA $ \mathcal{M} $, there exists a $ n^{2} $-state GFA $ \mathcal{G} $
	 such that $ f_{\mathcal{M}}(w)=f_{\mathcal{G}}(w) $
	for any $ w \in \Sigma^{*} $.\index{generalized finite automaton}\index{automaton!generalized}
\end{theorem}
\begin{proof}
	If we apply a superoperator, say $ \mathcal{E} $, to a quantum system in a state $\rho$,  the new state is
	$ \rho^{\prime} = \mathcal{E}(\rho) = \sum_{i=1}^k E_i \rho E_i^{\dagger}$.
	From this expression, one can see that the entries of the density matrix $\rho'$ are linear combinations of
	the entries of $ \rho $.
	
	We can linearize the computation of a given 1QFA 
	(with a quadratic increase in the size of the set of states \cite{YS11A}) 
	in a following way. We transform the density matrix into a real-valued 
    vector, replacing each complex entry of the density matrix with two real-valued elements
    of the vector. We choose the transition matrices $A_{\sigma}$ of the 1GFA so that
    they transform this vector in the same way as the superoperators $T_{\sigma}$ of the 1QFA
    transform the density matrix.
\end{proof}

Due to the equivalence 
between 1GFAs and 1PFAs, this simulation result is very useful. 
For example, it is used to show the equivalence of 1PFAs and 1QFAs in the unbounded error\index{unbounded error}\index{error!unbounded} case (Sec. \ref{sec:unbounded-error-realtime}) and various decidability\index{decidability results} and undecidability 
results (Sec. \ref{sec:decision-problems}).

In the bounded-error case, 1QFAs can recognize only regular languages 
(similarly to 1PFAs).\index{bounded error}\index{error!bounded}
The pure state version of this result was first shown for KWQFA in \cite{KW97}, 
with the bounds on the number of states shown in \cite{AF98}. 

\begin{theorem}
	\label{thm:1QFA-purestate-to-rtDFA}
	 \cite{KW97,AF98}
	If a language $L$ is recognized by an $ n $-state 1QFA with \textit{pure states} 
(e.g. MCQFAs and 1KWQFAs) with bounded error, then it can be recognized by a 1DFA with $2^{O(n)}$ states.
\end{theorem}
\begin{proof}
	Let $\cal M$ be the minimal 1DFA that recognizes $L$ and 
	let $N$ be the number of states of $\cal M$. 
	Let $q_1$ and $q_2$ be two states of $\cal M$. Then, there is a string $w$ such that
	reading $w$ in one of states $q_1, q_2$ leads to an accepting state and
	reading it in the other state leads to a rejecting state. 
	
	Let $\cal M'$ be a 1QFA that recognizes $L$. Let $w_1, w_2$ be strings after reading which
	$\cal M$ is in states $q_1$, $q_2$ and $\ket{\psi_1}, \ket{\psi_2}$ be the states
	of $\cal M'$ after reading $w_1, w_2$. 
	Let $T$ be the sequence of transformations
	that corresponds to reading $w$ (including the final measurement that produces the
	answer that says whether $\cal M'$ accepts or rejects the input string $w$). 
	If $\cal M'$ correctly recognizes $L$, then applying $T$ to one of $\ket{\psi_1}, \ket{\psi_2}$
	leads to a ``yes" answer with probability at least 2/3 and applying $T$ to the other state
	leads to a ``no" answer with probability at least 2/3. The next lemma provides a    
	necessary condition for that.
\begin{lemma}
\cite{BV97}
\label{lem:dist1}
Let 
\[ \ket{\psi_1} = \sum_{i=1}^n \alpha_i \ket{i} \mbox{~~~and~~~}
\ket{\psi_2} = \sum_{i=1}^n \beta_i \ket{i} .\]
Then, for any $~T$, the probabilities of $~T$ producing ``yes" answer 
on $\ket{\psi_1}$, $\ket{\psi_2}$
differ by at most $\|\psi_1-\psi_2\|$ where 
\[ \|\psi_1 -\psi_2\| = \sqrt{\sum_{i=1}^n |\alpha_i -\beta_i|^2} .\]
\end{lemma}
Hence, if a 1-way pure state QFA $\cal M'$ with $n$-dimensional state space recognizes $L$,
there must be $N$ pure states $\ket{\psi_1}, \ldots, \ket{\psi_N}$ 
in $n$ dimensions such that $\|\psi_i-\psi_j\|\geq 1/3$ for all $i, j:i\neq j$.
Such sets of states are known as quantum fingerprints\index{quantum fingerprints} and quite
tight bounds for the maximum number of quantum fingerprints in $n$ 
dimensions are known \cite{BCWW01}.
In particular, we know that $N=2^{O(n)}$.
\end{proof}

We now sketch the proof of a similar result for the general case.
Simulations of general 1QFAs by DFAs can be found in several papers (for example, 
\cite{LQZLWM12}) but we also provide an upper bound on the number of states.

\begin{theorem}
	\label{thm:1QFA-mixedstate-to-rtDFA}
	If a language is recognized by an $ n $-state 1QFA with \textit{mixed states} (e.g. 1QFA) with bounded error,\index{bounded error}\index{error!bounded} then it can be recognized by a 1DFA with $2^{O(n^2)}$ states.
\end{theorem}
\begin{proof}
	The proof is similar to the previous theorem but now we have to answer the
	question: how many mixed states 
    $\rho_i$ can one construct so that, for any $i\neq j$,
	there is a sequence of transformations $T$ that produces different outcomes (``yes" in    
    one case and ``no" in the other case) with probability at least 2/3?
 
	The answer is that the number of such $\rho_i$ in $n$ dimensions is at most $2^{O(n^2)}$.
		This follows from the fact that a mixed state in $n$ dimensions 
		can be expressed as a mixture $(p_l, \ket{\psi_l})$
		of at most $n$ pure states. We can then approximate each of 
		$\ket{\psi_l}$ by a state $\ket{\psi'_l}$ 
		from an $\epsilon$-net for the unit sphere in $n$ dimensions. 
(An $\epsilon$-net is a set of states $S$ such that, for any $\ket{\psi}$, there exists $\ket{\psi'}\in S: \|
		\psi-\psi'\|\leq\epsilon$.)
		Since one can construct an $\epsilon$-net with $2^{O(n)}$ states,
		there will be $(2^{O(n)})^n = 2^{O(n^2)} $ choices for the set of states
		$( \ket{\psi_1}, \ket{\psi_2}, \ldots, \ket{\psi_n})$.
		We also need to use another $\epsilon$-net for $(p_1, \ldots, p_n)$ but the size of this
		$\epsilon$-net is $2^{O(n)}=2^{o(n^2)}$. 
\end{proof}

\subsection{Succinctness results} 
\label{sec:succinctness}
\index{state complexity}

In \cite{AF98}, it was shown that 1QFAs can indeed be exponentially more succinct  than 1PFAs.
Let $ p $ be a prime  and consider the language $\mathtt{PRIME_p}$ defined by (\ref{eq:language:MOD-p}).\index{prime language}

\begin{theorem}
	\cite{AF98}
	\label{theorem:1QFA-succinctness-AF98}
	(i)
	If $p$ is a prime, any 1PFA recognizing $ \mathtt{MOD_p} $ has at least $p$ states.
	(ii)
	For any $ \epsilon>0 $, there is a MCQFA with $ O(\log(p)) $ states recognizing $ \mathtt{MOD_p} $ with error bound $ \epsilon $.
\end{theorem}

We now describe the construction of \cite{AF98} (in a simplified form due to \cite{AN09}). 
Let $\mathcal{M}_{k}$, for $k\in\{1, \ldots, p-1\}$, be the two state MCQFA given in Lemma 
\ref{lem:2-state-MCQFA} for $ \phi = \frac{2 \pi k }{p} $.
Thus, $\mathcal{M}_{k}$ accepts $a^j$
with probability $\cos^2 (\frac{2\pi j k}{p})$.
If $p$ divides $j$, then $\cos^2(\frac{2\pi j k}{p})=\cos^2 0 = 1$.
For $j$ that are not divisible by $p$, about half of them are accepted with 
probability less than $ \frac{1}{2} $. (This happens if $\frac{2\pi j k}{p}$
belongs to one of the intervals $[2\pi m + \frac{\pi}{4}, 2\pi m + \frac{3\pi}{4}]$
or $[2\pi m + \frac{5\pi}{4}, 2\pi m + \frac{7\pi}{4}]$.)
That is, each of the MCQFAs $\mathcal{M}_{k}$ distinguishes $a^j\in L_p$ from many (but not all) $a^j\notin L_p$. We now combine $O(\log n)$ of the $\mathcal{M}_{k}$'s into one MCQFA $\mathcal{M}$ which distinguishes $a^j\in L_p$ from all $a^j\notin L_p$.

Let $k_1, \ldots, k_d$ be a sequence of $d$ numbers, for an appropriately chosen $d=O(\log p)$. 
The set of states of $\mathcal{M}$ consists of $2d$ states $q_{1,1}$, $q_{1,2}$,
$q_{2,1}$, $q_{2,2}$, $\ldots$, $q_{d,1}$, $q_{d,2}$.
The transformation for $a$ is defined by
\[ U_a(q_{i,1}) = \cos\frac{2k_i \pi}{p} \ket{q_{i,1}}+\sin\frac{2k_i \pi}{p} \ket{q_{i,2}} ,\]
\[ U_a(q_{i,2}) = -\sin\frac{2k_i \pi}{p} \ket{q_{i,1}}+\cos\frac{2k_i \pi}{p} \ket{q_{i,2}} .\]
That is, on states $ \ket{q_{i,1}}$, $ \ket{q_{i,2}}$, $\mathcal{M}$ acts in the same way as $\mathcal{M}_{k_i}$. 

The starting state is $\ket{q_{1, 1}}$. The transformation $U_{\cent}$ 
can be any unitary transformation that satisfies $ U_{\cent} \ket{q_{1,1}}=\ket{\psi_{0}} $
where $ \ket{\psi_{0}}=\frac{1}{\sqrt{d}} (\ket{q_{1,1}}+\ket{q_{2,1}}+ \cdots+ \ket{q_{d,1}}) $ and
$U_{\dollar}$ can be any unitary transformation that satisfies
$ U_{\dollar} \ket{\psi_{0}}=\ket{q_{1,1}} $.
The set of accepting states $ Q_{a} $ consists of one state $q_{1,1}$.

If $p$ divides $j$, then, by Lemma \ref{lem:2-state-MCQFA}, the transformation 
$(U_a)^j$ maps $\ket{q_{i, 1}}$ to itself.
Since this happens for every $i$, the state $\ket{\psi_{0}}$ is also
left unchanged by $(U_a)^j$. Thus, $ U_{\dollar} $ maps $\ket{\psi_{0}}$ to $\ket{q_{1, 1}}$, the only accepting state.

If $j$ is not divisible by $p$, we have the following result:

\begin{theorem}
\cite{AN09}
There is a choice of $ d = 2 \log \frac{2p}{\epsilon}$ values
$k_1, \ldots, k_d\in\{1, \ldots, p-1\}$ such that the MCQFA $\mathcal{M}$ described above
rejects all $a^j\notin L_p$ with probability at least $1-\epsilon$.
\end{theorem} 
  
The proof of this theorem is nonconstructive: It was shown in \cite{AN09} that a random
choice $k_1, \ldots, k_d$ works with a high probability.
An explicit construction of $k_1, \ldots, k_d$ for a slightly larger $d=O(\log^{2+3\epsilon} p)$
also given in \cite{AN09}. 
Constructing an explicit set $k_1, \ldots, k_d$ such that $d=O(\log p)$ and $\mathcal{M}$ recognizes 
$ \mathtt{MOD_p} $ is still an open problem, which is linked to estimating exponential sums in number theory \cite{Bo05}. 

Currently, it is also open what is the biggest possible advantage of general 1QFAs over 1PFAs or 1DFAs.
By Theorem \ref{thm:1QFA-mixedstate-to-rtDFA}, 1QFAs with $n$ states can be simulated by 1DFAs with $2^{O(n^2)}$
states. On the other hand, \cite{Fre08,FOM09} gives a 1QFA with $n$ states 
for a language in an $2^{\Omega(n \log n)}$ symbol alphabet
that requires $2^{\Omega(n \log n)}$ states on 1DFAs and 1PFAs.
(The paper \cite{Fre08} also claims a similar result for a language in a 4-symbol
alphabet but the proof of that appears to be either incomplete or incorrect.)

There has been a substantial amount of further work on the state complexity\index{state complexity} of 1QFAs.
We highlight some results:
\begin{enumerate}
	\item \cite{MP02} Any periodic language \index{periodic language}\index{language!periodic} with period $ n $ in a unary alphabet\index{unary language}\index{language!unary} can be recognized by a $ (2\sqrt{6n}) $-state MCQFA with bounded error.
	\item \cite{BMP03B} There exists a language with period $ n $ in a one-symbol alphabet that requires $ \Omega(\sqrt{\frac{n}{\log n}}) $ states to be recognized by MCQFAs.
	\item \cite{BMP05A} If the $l_1$ norm of the Fourier transform of the characteristic function of $L$ (for a periodic $L$ in one-symbol alphabet) is small, a 1QFA with a smaller number of states is possible. 
\end{enumerate}

There are also some negative results for restricted models. For example, there is a language that is recognized by a 1DFA with $n$ states but it requires $2^{\Omega(n)}$ states for NaQFAs  \cite{ANTV02}. Due to Theorem \ref{thm:1QFA-purestate-to-rtDFA}, no such result is possible
for the more general models. 

\subsection{Bounded-error language recognition in restricted models}
\index{bounded error}\index{error!bounded} 
For language recognition with bounded error, we can put the models of QFAs in the order from the weakest to the strongest:
\begin{equation}
	\mbox{MCQFA} < \mbox{LaQFA} < \mbox{KWQFA} \leq  \mbox{NaQFA} \leq \mbox{BiQFA}
	< 
	\begin{array}{l}
		\mbox{1QFA} \\
		\mbox{QFA-CL} \\
		\mbox{QFA-A} \\
		\mbox{CiQFA} \\
		\mbox{1QCFA}
	\end{array} \equiv \mbox{1DFA} .
\end{equation}
It is open whether the inclusions KWQFA$\leq$ NaQFA and NaQFA$\leq$ BiQFA are strict.

The class of languages recognized by MCQFAs with bounded error ($ \mathsf{RMO} $)
is exactly the class of group languages\index{group language}\index{language!group} \cite{BC01A,BP02,MC00}.
(See \cite{Pi87} for the definition and the details about group languages.)
Belovs et al. \cite{BRS07} have shown that the power of MCQFAs can be increased by
reading multiple symbols at a time. However, even in this case, they cannot
recognize all regular languages \cite{BRS07,QY09}.

Similar to MCQFAs, a complete characterization of the class of languages recognized by LaQFAs ($\mathsf{BLaQAL}$) was obtained by algebraic techniques \cite{ABGKMT06}.
Namely, $\mathsf{BLaQAL}$ is equal to the class whose syntactic monoid is in $ \mathsf{BG} $, i.e. block groups\index{block groups}. Therefore, $\mathsf{BMO}\subset \mathsf{BLaQAL}$.  On the other hand, $\mathsf{BLaQAL}$ is a proper subset of the class of languages recognized by KWQFAs with bounded error, i.e. $ \mathsf{BMM} $, since $ \mathsf{BMM} $ contains $ \{ a\{a,b\}^{*} \} $ \cite{ABGKMT06} which is not in $\mathsf{BLaQAL}$. 

The classes of languages recognizable by other models of 1QFAs 
have not been characterized so well (and it is not clear whether they even have simple characterizations).
While 1KWQFAs recognize more languages than MCQFAs and LaQFAs, 
they cannot recognize some regular languages \cite{KW97}, for example, $ \{ \{a,b\}^{*}a \} \notin \mathsf{BMM} $.

Researchers \cite{AF98,BP02,GKK11} have also shown that if a minimal 1DFA has some certain properties, called  \textit{forbidden constructions}, then its language cannot be recognized by KWQFAs for any or for some error bounds. Using forbidden constructions, it was also shown that $ \mathsf{BMM} $ is not closed under intersection or union \cite{AKV01}. Moreover, KWQFAs can recognize more languages if the error bound gets closer to $ \frac{1}{2} $ \cite{ABFK99}.

NaQFAs and BiQFA share many of the properties of 1KWQFAs in the bounded error setting \cite{Mer08,GKK11}. In \cite{GKK11}, it was shown that any language recognized by a BiQFA with bounded error is in the language class $ \mathsf{ER} $ which is a proper subset of REG. 

The relative power of various models has also been studied for subclasses of
regular languages.
For unary regular languages,\index{unary language}\index{language!unary} 
LaQFAs (and all models that are more powerful than LaQFAs) recognize all unary regular languages since all unary regular
languages are in $ \mathsf{BG} $ language variety \cite{Mar12}.
MCQFAs cannot recognize all unary regular languages because 
they cannot recognize any finite language. 

If we restrict ourselves to $ \mathsf{R_1} $ languages, another proper subclass of regular languages, the computational powers of KWQFA, NaQFA, and BiQFA become equivalent \cite{GKK11}.

\subsection{Unbounded-error, nondeterminism, and alternation} 
\label{sec:unbounded-error-realtime}

In the unbounded error\index{unbounded error}\index{error!unbounded} setting, the language recognition power of 1PFAs and 1QFAs (and GFAs) are equivalent. This result is followed by combining Theorems \ref{thm:1PFA-simulated-by-1QFA} and \ref{thm:1QFA-is-a-GFA} and the simulation of GFAs by 1PFAs given in \cite{Tur69}. That is, due to Theorem \ref{thm:1PFA-simulated-by-1QFA}, any language recognized by a 1PFA with cutpoint 
is also recognized by a 1QFA with cutpoint; due to Theorem \ref{thm:1QFA-is-a-GFA}, any language recognized by a 1QFA is also recognized by a GFA; and any language recognized by a GFA with cutpoint is stochastic\index{stochastic language} \index{language!stochastic} \cite{Tur69}. Therefore, the class of languages recognized by 1QFAs with unbounded error is $ \mathsf{uS} = \mathsf{S} \cup \mathsf{coS} $ \cite{YS11A}. Note that 
it is still an open problem \textit{whether $ \mathsf{S} $ is closed under complementation (page 158 of \cite{Pa71})}.

Unlike in the bounded error case, KWQFAs are sufficient  
to achieve equivalence with 1PFAs in the unbounded error setting \cite{YS09C}. For weaker models
of QFAs, MCQFAs can recognize a proper subset of $ \mathsf{uS} $ with unbounded error \cite{MC00,BC01B} and
the class of languages recognized by MCQFAs with cutpoint is not closed under complement. Moreover, if a language is recognized by a MCQFA and the isolation gap is $ \frac{1}{p(n)} $ for some polynomial $ p $, then it is a regular language \cite{BC01B}. For LaQFAs, it is still an open problem whether they can recognize every stochastic language\index{stochastic language}\index{language!stochastic} with a cutpoint.

Nondeterministic version of quantum models are defined by fixing the error type to positive one-sided unbounded error\index{one-sided unbounded error}\index{unbounded error!one-sided}\index{error!one-sided unbounded} 
\cite{ADH97}. That is, all strings accepted with non-zero probability forms the language recognized by the nondeterministic quantum model. Remark that the same definition also works for classical models.

1NQFAs\index{nondeterministic quantum automaton}\index{quantum automaton!nondeterministic}\index{automaton!nondeterministic quantum} are the nondeterministic version of 1QFAs. $ \mathsf{NQAL} $ denotes the class of languages recognized by 1NQFAs \cite{YS10A}. The first result on the power of 1NQFAs was that 1NQFAs can recognize some nonregular languages such as $ \mathtt{NEQ}=\{ w \in \{a,b\}^{*} \mid |w|_{a} \neq |w|_{b} \} $ \cite{BC01B,BP02}.
A complete characterization of $ \mathsf{NQAL} $ was given in \cite{YS10A}:  $ \mathsf{NQAL} = \mathsf{S}^{\neq} $.

Similarly to the unbounded-error case, the most restricted 1NQFA model recognizing all languages in $ \mathsf{S}^{\neq} $ is nondeterministic KWQFAs. Moreover, since any unary language in $ \mathsf{S}^{\neq} $ is regular (Page 89 of \cite{SS78}), 1NQFAs and 1NFAs have the same computational power on unary languages.\index{unary language}\index{language!unary} 

Setting error type to \textit{negative} one-sided unbounded error 
($M$ must accept all $x\in L$ with probability 1 and reject every $x\notin L$ with a non-zero probability), we obtain one-way universal QFAs (1UQFAs). A language $L$ is recognized by a 1UQFA if and only if its complement is recognizable by a 1NQFA.

Recently, alternating quantum models\index{alternating quantum automaton}\index{quantum automaton!alternating}\index{automaton!alternating quantum} were introduced as a generalization of nondeterministic quantum model \cite{Yak13B} and it was shown that one-way alternating QFAs with $ \varepsilon $-moves\footnote{The automaton can spend more than one step on each symbol.} can recognize any recursively enumerable language.\index{recursively enumerable language}\index{language!recursively enumerable} Their one-way variants are also powerful: they can recognize NP-complete problem\index{NP-complete problem} $ \mathtt{SUMSETSUM}  $ and some nonregular and nonstochastic unary languages\index{unary language} like $ \{ a^{n^2} \mid n \geq 0 \} $ with only two alternations and PSPACE-complete problem $\tt SUBSETSUM\mbox{-}GAME$ with unlimited alternation \cite{Yak16A,DHRSY14A}.

\subsection{Decidability and undecidability results} 
\label{sec:decision-problems}
\index{decidability results} \index{undecidability}
In this section, we consider decidability and complexity of various problems involving one-way QFAs whose transitions are defined using computable numbers
or a subset of computable numbers (e.g. rational or algebraic numbers). 

\subsubsection{Equivalence and minimization.}\index{equivalence of automata}\index{minimization of automata} Two automata $ \mathcal{A}_1 $ and $ \mathcal{A}_2 $ are said to be equivalent if $ f_{\mathcal{A}_1}(x) = f_{\mathcal{A}_2}(x) $ for all input strings $ x \in \Sigma^{*} $ and they are said to be $ l $-equivalent if $ f_{\mathcal{A}_1}(x) = f_{\mathcal{A}_2}(x) $ for all input strings $ x \in \Sigma^{*} $ of length at most $ l $.  It is known that \cite{Pa71,Tze92} any two GFAs $ \mathcal{G}_1 $ and $ \mathcal{G}_1 $ with $ n_{1} $ and $ n_{2} $ states are equivalent if and only if they are $ (n_1 + n_2 -1) $-equivalent.\footnote{The method presented in \cite{Tze92} was given for 1PFAs but it can be easily applied to any linearized one-way computational model (see also bilinear machine given in \cite{LQ08}).}  Due to Theorem \ref{thm:1QFA-is-a-GFA}, any $ n $-state 1QFA can be converted to an equivalent $ n^2 $-state GFA. Therefore, it follows that two 1QFAs $ \mathcal{M}_{1} $ and $ \mathcal{M}_{2} $ 	with $ n_{1} $ and $ n_{2} $ states are equivalent, if and only if they are $ (n_{1}^{2}+n_{2}^{2}-1) $-equivalent (see also \cite{BP02,LQ08,LQ09}.) Since there is a polynomial time algorithm for checking the equivalence of two GFAs with rational amplitudes \cite{Tze92}, this implies that the equivalence of two 1QFAs with rational amplitudes can be checked in polynomial time.

The minimization\index{minimization of automata} of a given 1QFA with algebraic numbers is decidable: there is an algorithm that takes a 1QFA as an input and then outputs a minimal size 1QFA that is equivalent to $ \mathcal{A} $ \cite{MQL12}. Moreover, the algorithm runs in exponential space if the transitions of 1QFAs are rational numbers \cite{Mat12,Qiu12}. 
In \cite{BMP06}, the problem of finding the minimum MCQFA for a unary periodic language\index{periodic language}\index{language!periodic} (given by a vector that describes which strings belong to the language) was studied
and it was shown that the minimum MCQFA can be constructed in exponential time.

\subsubsection{Emptiness problems and problems regarding isolated cutpoint.}
\index{emptiness problem} 
We continue with five emptiness problems and two problems regarding isolated cutpoints.  
Let $ \mathcal{A} $ be an automaton and $ \lambda $ be a cutpoint. 
\begin{enumerate}
	\item \label{item:fa-emptiness-greaterandequal} 
		Given $ \mathcal{A} $ and $ \lambda \in \left[ 0,1 \right]  $, is there any $ w \in \Sigma^{*} $, $ f_{\mathcal{A}}(w) \geq \lambda $?
	\item \label{item:fa-emptiness-lessandequal}
		Given $ \mathcal{A} $ and $ \lambda \in \left[ 0,1 \right]  $, is there any $ w \in \Sigma^{*} $, $ f_{\mathcal{A}}(w) \leq \lambda $?
	\item \label{item:fa-emptiness-equal}
		Given $ \mathcal{A} $ and $ \lambda \in \left[ 0,1 \right]  $, is there any $ w \in \Sigma^{*} $, $ f_{\mathcal{A}}(w) = \lambda $?
	\item \label{item:fa-emptiness-greater}
		Given $ \mathcal{A} $ and $ \lambda \in \left[ 0,1 \right]  $, is there any $ w \in \Sigma^{*} $, $ f_{\mathcal{A}}(w) > \lambda $?
	\item \label{item:fa-emptiness-less}
		Given $ \mathcal{A} $ and $ \lambda \in \left[ 0,1 \right]  $, is there any $ w \in \Sigma^{*} $, $ f_{\mathcal{A}}(w) < \lambda $?
	\item Given $ \mathcal{A} $ and $ \lambda \in (0,1) $, is the cutpoint isolated?
	\item Given $ \mathcal{A} $, is there a cutpoint $ \lambda' $ which is isolated?
\end{enumerate}
All of these problems are known to be undecidable for 1PFAs with rational-valued transitions and rational cutpoints \cite{Pa71,Be75,BMT77,BC03}. Since 1PFAs are a restricted form of 1QFAs, it follows that Problems 1-6 are undecidable for general 1QFAs with rational-valued transitions and rational cutpoints, and Problem 7 is undecidable for general 1QFAs with rational-valued transitions. 

On the other hand, the situation is not straightforward for the restricted one-way QFAs (\cite{BJKP05,DJK05}): Problems 1-3 are undecidable for MCQFAs with rational-valued transitions and rational cutpoints. However, Problems 4-6 are decidable for MCQFAs with algebraic-valued transitions and algebraic cutpoints, and Problem 7 is decidable for MCQFAs with algebraic-valued transitions.

Furthermore, Problems 1-3 remain undecidable for $ 13 $-state MCQFAs with rational-valued transitions  
and an input alphabet of size 7 \cite{DJK05} and for $25$-state MCQFAs with rational-valued transitions and binary alphabet \cite{Hir07}. If algebraic transitions are allowed, the number of states in undecidability results for MCQFAs can be decreased by $6$.

For KWQFAs, Problems 4 and 5 are undecidable for algebraic-valued transitions and rational cutpoints due to \cite{Jea02,YS09C,YS11A,Jea12}: given a fixed rational cutpoint $ \lambda \in [0,1] $, a 1PFA $ \mathcal{P} $ 
with algebraic-valued transitions can be transformed into a KWQFA $ \mathcal{M} $ with algebraic-valued transitions so that for any string $ x $, if $ f_{\mathcal{P}}(x) < \lambda $, $ f_{\mathcal{P}}(x) = \lambda $, or $ f_{\mathcal{P}}(x) > \lambda $, then $ f_{\mathcal{M}}(x) < \lambda $, $ f_{\mathcal{M}}(x) = \lambda $, or $ f_{\mathcal{M}}(x) > \lambda $, respectively. Therefore, the undecidability results for 1PFAs imply similar undecidability results for KWQFAs regarding Problems 4 and 5  \cite{Jea02,Jea12}.

Currently, it is open whether Problems 4 and 5 are decidable for LaQFAs and whether Problems 6 and 7 are decidable for the models between MCQFAs and general 1QFAs.

In \cite{CL89}, a promise version of the emptiness problem\index{emptiness problem} for 1PFAs was considered, with a promise that  
either the automaton accepts at least one input $w\in\Sigma^*$ with probability at least $ 1 - \epsilon $ or the accepting probability is at most $ \epsilon $ for all $w\in\Sigma^*$, where $ \epsilon < \frac{1}{2} $.
It was shown that this problem is undecidable for 1PFAs with rational-valued transitions. 
This implies that the same problem is also undecidable for 1QFAs.

Recently, the emptiness problem for alternating QFAs, \textit{i.e. whether the given automaton defines an empty set or not}, was examined in \cite{DHRSY14A},
with the following results: (i) the problem is decidable for NQFAs with algebraic-valued transitions on general alphabet and UQFAs with computable-valued transitions on unary alphabets,\index{unary language}\index{language!unary} but, (ii) it is undecidable for UQFAs on general alphabets and alternating 1QFAs on unary alphabets, where both of them are defined with rational-valued transitions.

\subsubsection{Other problems.}
In \cite{Pa71} (Theorem 6.17 on Page 190), the problem of deciding \textit{whether the stochastic language\index{stochastic language}\index{language!stochastic} recognized by a 1PFA $ \mathcal{P} $ with cutpoint $ \lambda $ is regular (or context-free)} was shown to be undecidable for 1PFAs with rational-valued transitions and a rational $ \lambda \in [0,1) $. By the discussion above, the same problem is undecidable for 1QFAs with rational-valued transitions and rational cutpoints and KWQFAs with algebraic-valued transitions and rational cutpoints.

A $ k $-QFA classifier \cite{BMP06} is a system of $ k $ QFAs $(\mathcal{M}_1,\ldots,\mathcal{M}_k)$ on an alphabet $ \Sigma $  such that each $ \mathcal{M}_i $ accepts at least one string with probability bigger than $ \frac{1}{2} $ and there is no string which is accepted by both $ \mathcal{M}_i $ and $ \mathcal{M}_j $ (for some $i, j: i \neq j $) with probability bigger than $ \frac{1}{2} $. A complete $ k $-QFA classifier is a $ k $-QFA classifier such that each string is accepted by exactly one QFA with probability bigger than $ \frac{1}{2} $. It was shown that \cite{BMP06} for any $ k \geq 2 $, it is decidable \textit{whether $(M_1, \ldots, M_k)$ is a $ k $-QFA classifier}. On the other hand, it is undecidable whether $(M_1, \ldots, M_k)$ is a complete $k$-QFA classifier. 

In \cite{BP10},  two polynomial-time algorithms were given for KWQFAs on a unary alphabet\index{unary language}\index{language!unary} $ \Sigma = \{a\} $ with rational-valued transitions. A KWQFA $ \mathcal{M} $ on a unary alphabet can be viewed as a quantum Markov chains.\index{quantum Markov chain}\index{Markov chain!quantum} Then, its non halting subspace decomposes into the ergodic and the transient subspaces (see \cite{AF98,BP10} for the details). The first algorithm of \cite{BP10} computes  the dimensions of these subspaces. The second algorithm decides whether $ f_{\mathcal{M}} $ has a period of $ d \geq 0 $ such that $ \forall k \in \mathbb{N} \left( f_{\mathcal{M}} (a^{k}) = f_{\mathcal{M}} (a^{k+d}) \right) $.

\section{Two-way QFAs}
\label{sec:2QFAs}
\index{two-way automaton}\index{automaton!two-way}
A two-way model has a read-only input tape on which the given input, say $ w $, is written between $ \cent $ (the left end-marker) and $ \dollar $ (the right end-marker) symbols. The tape square on which $ \tilde{w}_i $ is written is indexed by $ i $, where $ 1 \leq i \leq |\tilde{w}| $.

A two-way QFA\index{two-way quantum automaton}\index{quantum automaton!two-way}\index{automaton!two-way quantum} can be defined either as a fully quantum machine or a classical machine augmented with a finite-size quantum register (memory). The former one is known as two-way QFAs with quantum head 
(2QFAs), and the latter one is known as two-way QFAs with classical head (2QCFAs).

\subsection{2-way QFAs with classical head}

A 2-way QFA with classical head (2QCFA), also known as two-way finite automaton with quantum and classical states \cite{AW02}, is a 2-way automaton augmented with a quantum register. The computation is governed classically. In each step, the classical part applies a quantum operator to the quantum register and then updates itself by also taking into account any measurement outcome obtained from the quantum part.

Formally, a 2QCFA $\cal M$ is an 8-tuple
$
	\mathcal{M}=(S,Q,\Sigma,\delta,q_{1},s_1,s_a,s_r),
$
where 
\begin{itemize}
\item
$ S $ is the set of states for the classical part 
and $Q$ is the set of basis states 
for the quantum part;
\item $ \delta $ is a transition function (consisting of $ \delta_c $ and $ \delta_q $ that governs the classical part and the quantum part of the machine, described in more detail below); 
\item
$ s_{1} \in S $ and $ q_{1} \in Q $ are the initial states for the classical and the quantum part, respectively; 
\item
$ s_a \in S $ and $s_r \in S $ ($s_a \neq s_r$) are the accepting and the rejecting states, respectively.
\end{itemize}

Each step of $ \mathcal{M} $ has two stages: a quantum transition ($ \delta_q $) and then a classical transition ($ \delta_c $):
\begin{itemize}
\item
The classical state $s\in S \setminus \{s_a,s_r\}$ and the input symbol $\sigma \in \tilde{\Sigma}$ determine an action $ \delta_q(s,\sigma) $ that is performed on the quantum register. This action can be a unitary transformation or a projective measurement.
\item
Then, the computation is continued classically. If $ \delta_q(s,\sigma) $ was a unitary transformation, then 
the classical transition $ \delta_c(s,\sigma) $ is an element of $ S \times \{ -1,0,+1\} $ specifying a new classical state and a movement of the tape head (left, stay, or right, respectively). If $ \delta_q(s,\sigma) $ is a measurement, 
the classical transition $ \delta_c(s,\sigma,\tau) $ is also an element of $ S \times \{ -1,0,+1 \} $ but is defined by a triple $(c, \sigma, \tau)$ which includes the outcome 
$\tau$ of the measurement on the quantum part.
\end{itemize}

At the beginning of the computation, the head is on the left end-marker, the classical state is $ s_1 $, and the quantum state is $ \ket{q_1} $. The computation is terminated and the input is accepted (resp., rejected) when $ \mathcal{M} $ enters the state $ s_a $ (resp., $ s_r $). It is obvious that any 2PFA can be simulated by a 2QCFA.

A particular case of a 2QCFA is a 1QFA with restart: it reads the input from the left to the right in one-way mode, and if the computation does not halt (does not enter an accepting or rejecting state), the computation is restarted after reading the right end-marker \cite{YS10B}. Its probabilistic counterpart is 1PFA with restart.

\subsubsection{Bounded-error language recognition.}\index{bounded error}\index{error!bounded} Unlike one-way models, 2QCFAs are more powerful than their classical counterpart (2PFAs) \cite{AW02}:
\begin{itemize}
\item
the language $ \mathtt{EQ}  = \{ w \in {a,b}^* \mid |w|_a = |w|_b \}  $ can be recognized by 2QCFAs in polynomial expected time \cite{AW02} 
but can be  recognized by 2PFAs only in exponential time \cite{Fr81,DS90};
\item
the language $ \mathtt{PAL} = \{w \in \{a,b\}^{*} \mid w=w^{r} \} $ can be recognized by 2QCFAs in exponential expected time but cannot be recognized by 2PFAs (and more 
generally, by Turing machines with working tape of size $o(\log n)$) at all.
\end{itemize}

We now describe 2QCFAs for these languages. Both of them execute an infinite loop with two parts. The first part is quantum and the second part is classical.

\textbf{2QCFA for $ \mathtt{EQ} $:} The 2QCFA $\mathcal{M}_1$ has two quantum states $ \{q_1,q_2\} $. 
\begin{itemize}
	\item In the quantum part, $\mathcal{M}_1$ starts in state $\ket{q_1}$ in its quantum register and reads $w$ from left to right.
Each time when $\mathcal{M}_1$ reads $a$, it applies a rotation by an angle $ \sqrt{2}\pi $ in the real $ \ket{q_1} $-$\ket{q_2}$ plane 
in a counterclockwise direction. When $\mathcal{M}_1$ reads $b$, it applies a rotation by $ \sqrt{2}\pi $ in the clockwise direction.
When $\mathcal{M}_1$ arrives at the right end-marker, the quantum register is measured in computational basis and the input is rejected if $ \ket{q_2} $ is observed. 

If $w\in \mathtt{EQ} $, the rotations in both directions cancel out \ and the final quantum state is exactly $ \ket{q_1} $. Therefore, $w$ is never rejected. If $w\notin \mathtt{EQ}$, then the final quantum state is always away from $ \ket{q_1} $-axis and the resulting rejecting probability can be bounded from below by  $ \frac{1}{2|w|^2} $ (a nice property of rotation angle $\sqrt{2}\pi$). See Figure \ref{fig:2qcfa-for-eq} for some details of the quantum phase.
\begin{figure}[!htb]
\centering
\includegraphics[scale=0.4]{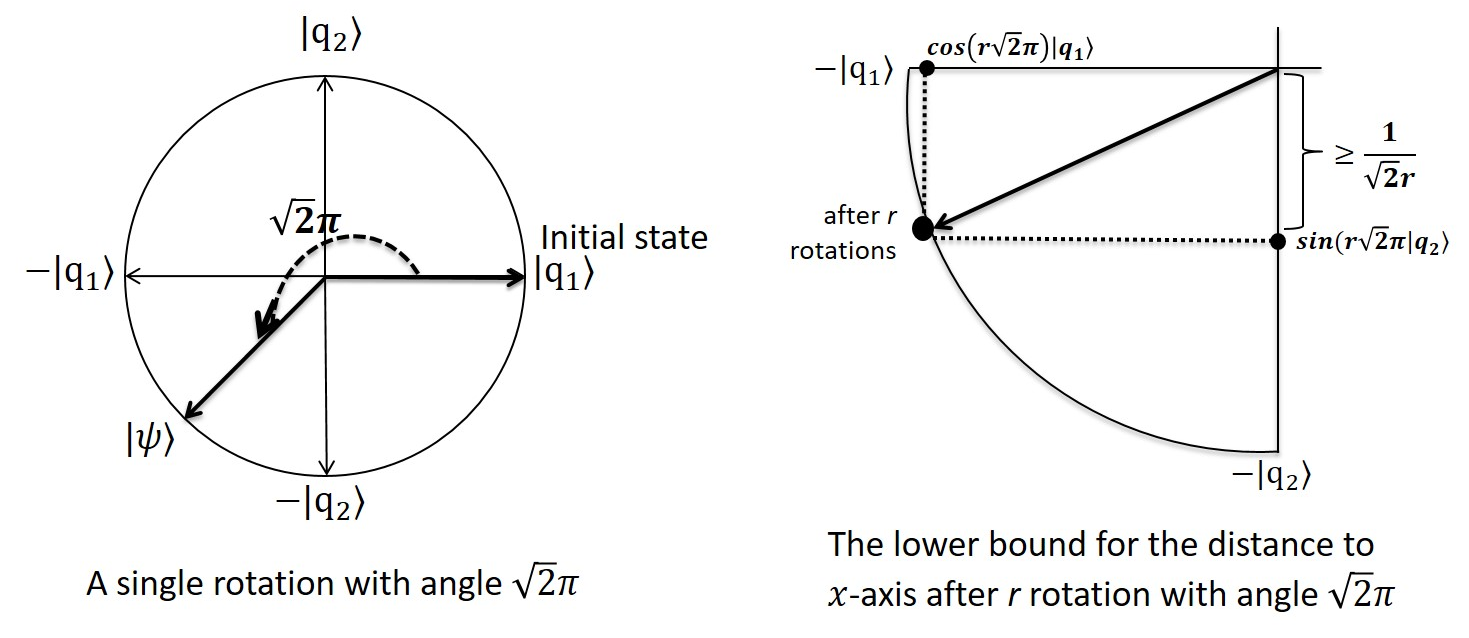}
\caption{Some details of the quantum phase, partially taken from \cite{SayY14A}}
\label{fig:2qcfa-for-eq}
\end{figure}
\item In the classical part, $ \mathcal{M}_1 $ performs a classical procedure (two consecutive random walks\index{random walk} on the string $w$ and then $k$ consecutive coin flips) that results in accepting $w$ with probability $ \frac{1}{2^k|w|^2} $ for some $ k>1 $, using expected time $O(|w|^2)$.
\end{itemize}
If $w\in \mathtt{EQ} $, $w$ is accepted with probability 1 in $ O(|w|^4) $ expected time, through the classical part of the loop.
If $w\notin \mathtt{EQ} $, the probability of rejection in the quantum part of the loop is larger than the probability of accepting in the classical part, i.e., $w$ is rejected with a probability at least $ \frac{2^{k}}{2^{k}+2} > \frac{1}{2} $ in $ O(|w|^2) $ expected time. \qed

\textbf{2QCFA for $ \mathtt{PAL} $:} The 2QCFA $ \mathcal{M}_2 $ has three quantum states $ \{q_1,q_2,q_3\} $. 
\begin{itemize}
	\item It starts the quantum phase in state $ \ket{q_1} $ and reads the input from left to right twice. In the first read, it applies 
\[
	U_a =  \dfrac{1}{5} \left( \begin{array}{rrr} 
		4 & 3 & 0 
		\\
		-3 & 4 & 0 
		\\
		0 & 0 & 5
\end{array}	 \right)
\mbox{~~ and ~~}
U_b =  \dfrac{1}{5} \left( \begin{array}{rrr} 
		4 & 0 & 3 
		\\
		0 & 5 & 0 
		\\
		-3 & 0 & 4
\end{array}	 \right)
\]
for each $a$ and $b$, respectively. In the second reading, it applies inverses of the matrices, respectively. Then, the quantum register is fully measured with respect $\ket{q_1}, \ket{q_2}, \ket{q_3}$ and the input is rejected if the result is not $\ket{q_1}$. So, if $w$ is palindrome, the state ends in $ \ket{q_1} $, i.e.
\[
	\ket{q_1} =  U_{w_{|w|}}^{-1} U_{w_{|w|-1}}^{-1} \cdots  U_{w_1}^{-1} U_{w_{|w|}} \cdots U_{w_2} U_{w_1} \ket{q_1},
\]
 and so $w$ is not rejected. Otherwise, the computation does not return to the initial quantum state exactly, which is away from $ \ket{q_1} $ by at least a value exponential small in the length of input, and the input is rejected with a probability at least $ 25^{-|w|} $ (due to the certain properties of $ U_a $ and $U_b$, see \cite{AW02} for the details).
	\item Similar to $ \mathcal{M}_1 $, in the classical phase the input is accepted with a sufficiently small probability, i.e. $ 2^{-4k|w|} $ for some $ k>1 $. 
\end{itemize}  
Thus, $ \mathcal{M}_2 $ accepts $ w $ with probability 1 if $w\in \mathtt{PAL} $ and rejects $ w $ with a probability at least $ \frac{16k}{16k+25} > \frac{1}{2} $,
otherwise.\qed

We note that $ \mathcal{M}_2 $ only uses rational-valued amplitudes. On the other hand, 
allowing arbitrary real numbers does not help 2PFAs for recognizing $ \mathtt{PAL} $ \cite{DS92}. 
 
These results have been generalized in \cite{YS10B}, by showing that all languages in  $ \mathsf{S}^{=}_{\mathbb{Q}} \cup \mathsf{S}^{\neq}_{\mathbb{Q}}  $\index{stochastic language} 
can be recognized by KWQFAs with restart (and so by 2QCFAs) with bounded error. S$ ^{=}_{\mathbb{Q}} $ contains many well-known languages: 
$ \mathtt{EQ} $, 
$ \mathtt{PAL} $, 
$ \mathtt{TWIN} =\{ wcw \mid w \in \{a,b\}^{*} \} $,
$ \mathtt{SQAURE}= \{a^{n}b^{n^{2}} \mid n > 0 \} $,
$ \mathtt{POWER} = \{ a^{n}b^{2^{n}} \mid n > 0 \} $, the word problem\index{word problem} of finitely generated free groups, all \textit{polynomial languages}\index{polynomial languages} defined in \cite{Tur82}, and 
$ \mathtt{MULT} =\{x \# y \# z \mid x,y,z \mbox{ are natural numbers in binary notation and } x \times y = z \} $.
Note that KWQFA with restart is the most restricted of known two-way QFA models that is more powerful than its classical counterpart (1PFA with restart).

\subsubsection{Succinctness results.} 
\index{state complexity} 2QFAs can also be more succinct  than their one-way versions and their classical counterparts \cite{YS10B,ZGQ14A}. The main result is that for any $ m>0 $, $ \overline{ \mathtt{MOD_{m}} } $ (the complement of language $\mathtt{MOD_{m}}$ defined in (\ref{eq:language:MOD-p})) can be recognized by a 1QFA with restart (and so by a 2QCFA) with a constant number of states for any one-sided error bound. 
On the other hand, the number of states required by bounded-error 2PFAs increases when $ m $  gets bigger. This also implies a similar gap between 2QCFAs and 1QFA: due to Theorem \ref{thm:1QFA-mixedstate-to-rtDFA}, a 1QFA with a constant number of states can be simulated by a 1DFA (and, hence, 2PFA) with a constant number of states (where the constant may be exponentially larger).

\subsubsection{Other results.} In \cite{RV07}, the simulation of a restricted bounded-error 2QCFA by weighted automata\index{weighted automaton}\index{automaton!weighted} was presented. No other ``non-trivial'' upper bound is known for bounded-error 2QCFAs. On the other hand, 
 it was shown that, if we allow arbitrary transition amplitudes (including non-computable ones), bounded-error 2QCFAs can recognize uncountably many languages in polynomial time \cite{SayY14B}. This is an evidence that 2QFAs can be very sensitive to the type of numbers that we use as transition amplitudes.

\subsection{2-way QFAs with quantum head}
The definition of 2QFAs with quantum head is  technically more difficult than that of 2QCFAs.
Because of that, we only provide an informal definition and an example of a 2QFA and  refer the reader to \cite{YS11A} for the remaining details. 

Let ${\cal M}$ be an $n$-state 2-way automaton with the set of states $\{ q_1,\ldots,q_n \}$
and let $w$ be an input string. Then, the possible configurations of ${\cal M}$ on the input $w$ can
be described by pairs $(q_i,j)$ consisting of automaton's internal state $i:1 \leq i \leq n$ and the location $j:1 \leq j \leq |\tilde{w}|$ in the input string which the automaton
is currently reading. A probabilistic automaton (2PFA) can be in a probability distribution of the classical configurations during its computation. A 2QFA can be in a quantum state 
with the basis states $\ket{q_i,j}$.
The evolution of a 2QFA governed by quantum operators (measurements, unitary operators, superoperators, etc.).

$\cal M$ evolves according to a transition rule which depends on the current state $q_i$
and the symbol $\tilde{w}_j$ at the current location. For example, if $ \cal M $ evolves unitarily, we have local transitions of the form
\begin{equation}
\label{eq:local}
	\ket{q,j} \rightarrow \sum_{q'\in\{q_1, \ldots, q_n\}, c\in\{-1, 0, 1\}} 
\alpha^{(q, \tilde{w}_j)}_{q',c} \ket{q',j+c}
\end{equation}
where $c\in\{-1, 0, 1\}$ corresponds to moving left, staying in place, or moving right and the transition amplitudes $\alpha^{(q, \tilde{w}_j)}_{q',c}$ depend on the state $q$ before the transition and the symbol $\tilde{w}_j$ that the automaton reads. 
By combining those transitions for all $q$ and $j$, we get an operator $ U_{\mathcal{M}}(w) $ that describes the evolution of the whole state space of $M$. This operator $ U_{\mathcal{M}}(w) $ must be unitary for any $w \in \Sigma^*$. This implies a finite list
of constraints on the amplitudes $\alpha^{(q, \tilde{w}_j)}_{q',c}$ in the local transition rules (\ref{eq:local}), known as the {\em well-formedness conditions} \cite{Yak11A,YS11A}. 

To stop the computation, we perform a partial measurement on QFA's quantum state after each application of $ U_{\mathcal{M}}(w) $, with respect to the partition of basis states into the set of accepting states $Q_a$, the set of rejecting states $Q_r$, and
the set of non-halting states $Q_n$. If the result is $Q_a$ (resp., $Q_r$), the computation is terminated and the input is accepted (resp., rejected). Otherwise, the computation is continued. 

The model above is the first 2QFA model, called two-way KWQFA (2KWQFA) \cite{KW97}. Although some interesting results obtained based on this model, it is still open whether 2KWQFAs can simulate 2PFAs. The Hilbert space can also be evolved by superoperators \cite{YS11A}, and then 2QFAs can simulate both 2QCFAs and 2PFAs exactly. 

If the head of a 2QFA is not allowed to move to left, then we obtain a 1.5-way QFA (1.5QFA). Here ``1.5'' emphases  that the head is quantum and so it can be in more than one position during the computation. 



\subsubsection{Bounded-error language recognition.}\index{bounded error}\index{error!bounded} As described above, 2QFAs and 1.5QFAs can be in a superposition\index{superposition} over different locations of the input tape instead of only being in a superposition of states. This enables them to use the length of input as a counter. We present a linear-time 1.5-way KWQFA for the language $ \mathtt{EQ} $ using this idea.\index{quantum automaton!1.5-way}\index{automaton!1.5-way quantum}

\textbf{1.5-way KWQFA for $ \mathtt{EQ} $:} 
Our automaton $\cal M$ has $ 5 $ states $ \{ q_1, q_2, q_w, q_a, q_r   \} $, with $q_1$ as the starting state. To determine whether the input should be accepted, we use the following measurement: if the computation is in a configuration containing $ q_a $ (resp., $ q_r $), then the input is accepted (resp., rejected). Otherwise, the computation goes on. 

The transitions are defined as follows\footnote{All transitions that are omitted below are not significant and so they can be arbitrary by guaranteeing that the related operator is unitary.}:
\begin{itemize}
\item
On the left end-marker, the starting state $ \ket{q_1} $ is mapped to $ \frac{1}{\sqrt{2}} \ket{q_1} + \frac{1}{\sqrt{2}} \ket{q_2}  $ and the head of $\cal M$ moves one square to the right;
\item
On symbol $a$, $\cal M$ performs the mapping: $\ket{q_1}\rightarrow\ket{q_w}$, 
$\ket{q_w}\rightarrow\ket{q_1}$, $\ket{q_2}\rightarrow\ket{q_2}$, staying in place if
the state after the transformation is $\ket{q_w}$ and moving to the right otherwise; 
\item
On symbol $b$, $\cal M$ performs the mapping: $\ket{q_2}\rightarrow\ket{q_w}$, 
$\ket{q_w}\rightarrow\ket{q_2}$, $\ket{q_1}\rightarrow\ket{q_1}$, staying in place if
the state after the transformation  is $\ket{q_w}$ and moving to the right otherwise; 
\item
On the right end-marker, $\cal M$ maps $ \ket{q_1} \rightarrow \frac{1}{\sqrt{2}} \ket{q_a} + \frac{1}{\sqrt{2}} \ket{q_r} $ and $ \ket{q_2} \rightarrow \frac{1}{\sqrt{2}} \ket{q_a} - \frac{1}{\sqrt{2}} \ket{q_r} $.
\end{itemize}
An example run of the machine is given in Figure \ref{fig:1-5qfa}, in which each arrow represents a single step and it is clear that after the second step the head places on the different squares of the tape until the end of the computation where they meet again and so they affect each other.

\begin{figure}[!htb]
\centering
\includegraphics[scale=.4]{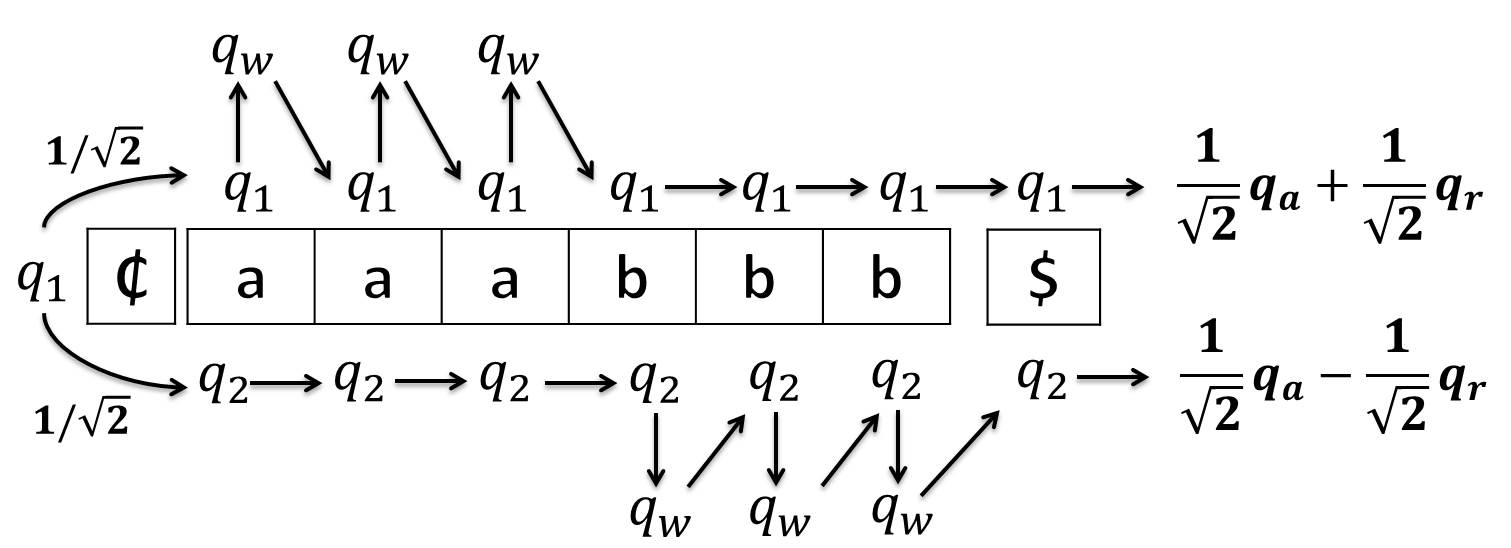}
\caption{An example run of $\mathcal{M}$}
\label{fig:1-5qfa}
\end{figure}

To analyze how $\cal M$ works, we observe that, on the left end-marker, it enters the state
$ \frac{1}{\sqrt{2}} \ket{q_1} + \frac{1}{\sqrt{2}} \ket{q_2}  $.
Every $a$ symbol results in the $q_1$ component moving to the right in 2 steps and 
the $q_2$ component moving to the right in 1 step. 
Every $b$ results in $q_1$ moving to the right in 1 step and 
$q_2$ moving to the right in 2 steps. 
If $|w|_a = |w|_b$, the automaton reaches the right end-marker at the same time in
$q_1$ and $q_2$. If $|w|_a \neq |w|_b$, one of components reaches the end-marker earlier than the other. 

In the first case, applying the transformation on the right end-marker gives the configuration 
\[
	\left( \frac{1}{2} \ket{ q_a,|\tilde{w}| } + \frac{1}{2} \ket{ q_r,|\tilde{w}| } \right) + 
    \left( \frac{1}{2} \ket{ q_a,|\tilde{w}| } - \frac{1}{2} \ket{ q_r,|\tilde{w}| } \right) =  \ket{ q_a,|\tilde{w}| } .
\]
So, the input is accepted with probability 1. In the second case,
we have $ \ket{q_1} \rightarrow \frac{1}{\sqrt{2}} \ket{q_a} + \frac{1}{\sqrt{2}} \ket{q_r} $ and $ \ket{q_2} \rightarrow \frac{1}{\sqrt{2}} \ket{q_a} - \frac{1}{\sqrt{2}} \ket{q_r} $ applied on the right end-marker
at different times and, in each case, $q_a$ and $q_r$ are obtained (observed) with equal probability.
Thus, the input is accepted with probability 1/2. \qed

The probability of accepting $x\notin  \mathtt{EQ} $ can be decreased from 1/2 to $1/k$, for
arbitrary $k$ \cite{KW97,YS09B}, with the number of states in the automaton increasing to
$O(k^2)$ using the construction of \cite{KW97} and to $O(\log^c k)$ using the construction of
\cite{YS09B}.

Currently we do not know any language separating 2QCFAs and 2QFA or any language requiring exponential expected time by two-way QFAs. Also, even though 1.5KWQFAs can recognize
non-regular languages (such as $\mathtt{EQ}$), it is not known whether they can recognize all regular languages with bounded error. It is also open whether 2QFAs can recognize a nonstochastic language\index{nonstochastic language}\index{language!nonstochastic} with bounded error.

\subsubsection{Unbounded-error language recognition.}
\index{unbounded error}\index{error!unbounded}
The superiority of 2QFAs also holds in the unbounded error  case. The language
\[ \mathtt{NH} = \{a^{x}ba^{y_{1}}ba^{y_{2}}b \cdots a^{y_{t}}b \mid x,t,y_{1}, \ldots, y_{t}  
	\in \mathbb{Z}^{+} \mbox{ and } \exists k ~ (1 \le k \le t), x=\sum_{i=1}^{k}y_{i} \}
\]
is nonstochastic\index{nonstochastic language}\index{language!nonstochastic} \cite{NH71} but is recognized by 1.5-way KWQFAs \cite{YS09C,YS11A} (by a generalization of the technique used by 1.5-way KWQFAs for $ \mathtt{EQ} $ in the previous section). This shows a superiority over probabilistic automata because 2PFAs cannot recognize any nonstochastic language \cite{Ka91}.
In fact, 1.5-way KWQFAs can recognize most of the nonstochastic languages defined in literature \cite{Die77,FYS10A}.
We note that the best known upper bound (in terms of complexity classes) for unbounded-error 2QFAs (with algebraic-valued transitions) is $ \mathsf{P} \cap \mathsf{L}^2 $ \cite{Wat03}. (Also see \cite{Yam15} for certain relations and upper bounds were defined on the running time of 2KWQFAs under different recognition modes.)

\subsubsection{Undecidability of emptiness problem.}\index{emptiness problem}\index{undecidability} 1.5KWQFAs have the capability of checking successive equalities, e.g. $ a^{n}ba^{n}ba^{n}ba^{n}\cdots $ (see \cite{Yak12A}). This leads to the following result: The emptiness problem for one-way KWQFAs (with algebraic-valued transitions) is undecidable \cite{AI99}. This is shown by a reduction from the halting problem for one-register machines, which is known to be undecidable. 

\section{Other models and results} 
\label{sec:other-models}

\textbf{Interactive proof systems.}\index{interactive proof system} An interactive proof system consists of two parties: the {\em prover} 
with an unlimited computational power and the {\em verifier} 
who is computationally limited.
Both parties are given an input $w$ and can send messages to one another.
We say that a language $L$ has an interactive proof system if there is a strategy for
the verifier with the following two properties:
\begin{enumerate}
\item[(a)]
If $w\in L$, there exists a strategy for the prover such that, given a prover who acts 
according to this strategy, the verifier accepts with probability at least $2/3$;
\item[(b)] 
If $w\notin L$, then, for any prover's strategy, the verifier rejects with probability 
at least $2/3$.
\end{enumerate}
$\sf QIP(M)$ denotes the class of all languages which have quantum 
interactive proof systems with verifiers of type $M$. Obviously, if $L$ is recognizable by type $M$ machine, it has a trivial interactive proof system in which the verifier runs its algorithm for recognizing $L$ and disregards the prover. Thus, $\sf QIP(M)$ can be much larger than $M$.

For finite automata, $\sf BMM$ is smaller than $\mathsf{REG}$ but $\mathsf{QIP(KWQFA)} = 
\mathsf{REG}$. For 2-way automata, we have the upper bound $ \mathsf{QIP_{\tilde
{\mathbb{C}}}(\mbox{poly-time } 2KWQFA)} \subseteq \mathsf{NP} $ \cite{NY09}.

In the multiprover 
version of this model (denoted $\sf QMIP(\cdot)$), the verifier can 
ask question to multiple provers and he is guaranteed that the provers do not interact 
one with another. Then, we know
\cite{Yam14}: 
\begin{itemize}
	\item $ \mathsf{CFL} \subseteq \mathsf{QMIP_{ \mathbb{ 
\tilde{C} }} (KWQFA)} \subseteq \mathsf{NE} $,
	\item $ \mathsf{QMIP_{ \mathbb{ \tilde{C} }} (\mbox{poly-time }2KWQFA)} = \sf NEXP $, and
	\item every recursively enumerable language\index{recursively enumerable language}\index{language!recursively enumerable} is in $ \mathsf{QMIP_{ \mathbb{ \tilde{C} }} (2KWQFA)} $,\end{itemize}
where $ \sf CFL $ and $ \sf NE $ are the classes of context-free languages\index{context-free language}\index{language!context-free} and 
languages recognizable in time $2^{O(n)}$.

It is interesting to compare this with the classical case where $\sf NEXP$ is equal to the 
class of all languages that have interactive proofs with a polynomial time Turing 
machine as the verifier (which is a much stronger model than a 2KWQFA).

An Arthur-Merlin\index{Arthur-Merlin} (AM) proof system is an interactive proof system in which all of the  verifier's probabilistic choices are visible to the prover. Thus, the prover has a complete information about the computational state of the verifier. In the quantum version, the verifier has a quantum register and the outcome is sent to the prover  whenever it is measured (so that the prover still has a complete information about the state of the verifier).

If the verifier is a 2QCFAs and
all the transitions are restricted with rational numbers, we have the following results 
\cite{Yak13C}: 
\begin{itemize}
	\item $   \mathsf{AM_{\mathbb{Q}}(2QCFA)} $ contains $ 
\mathsf{ASPACE(n)} \cup \mathsf{PSPACE} $ and some $\mathsf{NEXP}$-complete 
languages,\footnote{The proof of $   \mathsf{AM_{\mathbb{Q}}(2QCFA)} $ contains 
$\mathsf{PSPACE}$ will appear in an extended version of \cite{Yak13C}.} and,
	\item  
every recursively enumerable language\index{recursively enumerable language}\index{language!recursively enumerable} is in $ \mathsf{weak\mbox{-}AM_{\mathbb{Q}}(2QCFA)} 
$ where the prefix ``weak-'' denotes the class of languages having a proof system where 
the non-members do not need to be rejected with high probability.
\end{itemize}

The first result should be contrasted with the fact that 
$ \mathsf{weak\mbox{-}AM(2PFA)} $ is a proper subset of $ \mathsf{P} $ \cite{DS92}
and the second result should be compared with the fact that 
every recursively enumerable language\index{recursively enumerable language}\index{language!recursively enumerable} is in $ \mathsf{weak\mbox{-}IP(2PFA)} $ \cite{CL89}
(which is a similar result but uses a stronger computational model: IP instead of AM).

If we allow real and computable real numbers as amplitudes, $ \mathsf{AM_{\mathbb{R}}(2QCFA)} $ contains all languages and $ \mathsf{AM_{\dddot{\mathbb{R}}}
(2QCFA)} $ is equivalent to the class of recursive languages\index{recursive languages}\index{language!recursive} \cite{SayY14B}.
Moreover, it was shown that $ \mathsf{AM_{\mathbb{A}}(\mbox{poly-time }2QCFA)}$ contains a language that is not in $ \mathsf{AM_{\mathbb{R}}(\mbox{poly-time }2pfa)} $ \cite{ZQG15}.

Before closing this item, we also refer \cite{VilY14,NisY15} for further results on weaker QFA  verifiers in different set-ups. 

\textbf{Debate systems.} A debate system 
is a generalization of IP system where the verifier interacts with a prover (who tries to convince the verifier that the input $w\in L$) and a refuter 
(who tries to prove that the input $w\notin L$). If $w\in L$, there should be a strategy for the prover such that, regardless of the refuter's strategy, the verifier accepts 
with probability at least $2/3$. If $w\notin L$, the refuter should have a strategy such that, for any prover's strategy, the verifier rejects with probability at least $2/3$.

The debate version of $ \mathsf{AM_{\mathbb{Q}}(2QCFA)} $ has been shown to contain all recursive languages\index{recursive languages}\index{language!recursive} \cite{YSD14A}. In contrast, the debate version of $ \mathsf{AM_{\mathbb{Q}}(2PFA)} $ is a subset of $ \mathsf{NP} $ \cite{Con89}. 

\textbf{Postselection.} 
Postselection\index{postselection} is the ability to discard some outcomes at the end of the computation and to make the decision based on the surviving outcomes (even though these outcomes might be occurring with a very small probability). For example, if we have a QFA with 3 basis states
$\ket{q_1}, \ket{q_2}, \ket{q_3}$, we could discard the $\ket{q_3}$ part of the final state 
of the QFA and make the accept/reject decision based on the part of the final state which
consists of $\ket{q_1}$ and $\ket{q_2}$.

Postselection is not possible physically but is interesting as a thought experiment.  
It has been studied for both quantum circuits \cite{Aa05} and quantum automata \cite{SLF10,YS11B}. It has been shown that 1QFAs (1PFA) with postselection have the same computational power as 1QFAs (1PFAs) with restart.

\textbf{Closed Timelike curves.}\index{closed timelike curves}  
Similar to postselection, closed timelike curves (CTC) are a model which is impossible physically but is interesting as a thought experiment. A CTC is a device which allows to
send information back in time, to previous steps of the computation, as long as this does not result in inconsistencies in the computation. 

In \cite{SY11A,SY12B}, 1QFAs and 1PFAs with capability of sending one classical bit from the end of the computation to the beginning of the computation through a CTC have been examined. Surprisingly, it was shown that such 1QFAs can simulate 1QFAs with postselection, and vice versa, when their transitions are restricted to rational numbers. The same result was obtained also for 1PFAs even for arbitrary transition probabilities.

\textbf{Promise problems.} Promise problems\index{promise problems} are computational tasks where the goal is to separate two languages $L_1, L_2: L_1\cap L_2=\emptyset$ (the automaton must accept all
$w\in L_1$, reject all $w\in L_2$ and is allowed to output any answer for $w\notin L_1\cup L_2$). Promise problems allow to show separations between types of automata which are equivalent in the standard setting of recognizing languages. 

For example, for the case of exact computation\index{exact computation} (no error allowed), \index{error!zero} 1QFAs cannot be more concise than 1PFAs \cite{Kla00}. On the other hand, for promise problems, the superiority of 1QFAs over 1PFAs can be unbounded \cite{AY12}: There exists an infinite family of promise problems which can be solved exactly by tuning transition amplitudes of a two-state MCQFA, while the size of the corresponding classical automata grows to infinity \cite{RY14A,GY14A}. Recently, this result was generalized in \cite{GQZ14A} and \cite{BMP14B} and further succinctness results were given in \cite{ZQGLM13,GaiY15A,GQZ14B,ZGQ14A}. 

Several results about the computational power of QFAs on promise problems have been obtained  in \cite{RY14A,GaiY15A}. For example, there is a binary promise problem solvable by a Las Vegas\index{Las Vegas computation} 1QFA and a unary\index{unary language}\index{language!unary} promise problem solvable by a bounded-error 1QFA, but none of them can be solved by any bounded-error 1PFA. 
(For language recognition, these one-way models are of equal power and recognize exactly REG.) 
Moreover, there is a promise problem solvable by an exact 2QCFA in exponential expected time, but not by any bounded-error sublogarithmic space probabilistic Turing machine. No similar example is known for language recognition. Additionally, in \cite{ZLQG17}, a particular subset of promise problems solvable by one-way classical and quantum models was considered, and certain separation results were obtained.

\textbf{Advice.} In computation with advice\index{advice}, the automaton is provided extra information called {\em advice} which depends on the length of the input $w$ 
but not on the particular $w$. 

Advice is a well known notion in the complexity theory but has not been studied much in the setting of QFAs. The first model was introduced in \cite{Yam12} but was based on KWQFAs\footnote{Also, note that the usage of advice defined in \cite{Yam12} is different than the usual definition for classical finite automata \cite{DH95}.}. As a result, some regular languages were shown not to be recognized by this model, with advice of up to linear size. Recently, this framework was generalized in \cite{KSY13}, which can be a good starting point for studying QFAs with advice.

\textbf{Determining the bias of a coin.} \index{state complexity}
In \cite{AD11}, the state complexities of 1QFAs and 1PFAs were compared for the problem of determining the bias of a coin, if it is known that the coin lands ``heads'' with probability either $ p $ or $ p + \epsilon $ for some known $ p $ and $ \epsilon $. 
A 1QFA can distinguish between the two cases with a number of states that is independent of
$p$ and $\epsilon$ while any bounded-error 1PFA must have $ \Omega\left(\frac{p(1-p)}{\epsilon}\right) $ states \cite{AD11}. Recently, it was also proven that \cite{KinO17} there is no 1QFA having the following property: simultaneously for every $ \epsilon \in \left[ -\frac{1}{2},\frac{1}{2} \right] \setminus \{0\} $, given access to an infinite sequence of coin tosses, if the coin is $ \left( \frac{1}{2} + \epsilon \right) $-biased then the automaton spends at least $ 2/3 $ of its time guessing ``biased'', and if the coin is fair then the automaton spends at least 2/3 of its time guessing ``fair''.

\textbf{Learning theory.} The problem of learning probability distributions produced by QFA sources, i.e. identifying an unknown QFA from examples of its behavior, was studied in  \cite{Jub12}. Information-theoretically, QFAs can be learned  from a polynomial number of examples, similarly to classical hidden Markov models. However, computationally, the problem is as hard as learning noisy parities, a very difficult problem in computational learning theory \cite{BFKL93}.

\section{Concluding remarks } 
\label{sec:conclusion}

Quantum finite automata (QFAs) combine the theory of finite automata with quantum computing. Many different models and aspects of QFAs have been studied and this research topic has recently celebrated its 20 years. 

There are some contexts in which quantum models are of the same power as classical models (for example, language recognition power of 1QFAs with bounded or unbounded error) or
have similar properties as classical models (for example, undecidability of the emptiness problem for 1-way automata). On the other hand, there are many cases in which quantum models are superior to classical models (for example, succinctness results for almost all models, nondeterministic language recognition power, and language recognition power of 2QFAs with bounded or unbounded error). Besides these, there are still many research questions that are still open.

Among restricted one-way QFAs, LaQFAs deserve a special attention. Moreover, it would be interesting to find more examples where QFAs can be substantially smaller than DFAs and PFAs. So far, most examples are periodic languages over unary alphabet (e.g. \cite{AF98,MP02}) or their simple generalizations. This raises a question: for what non-unary languages do QFAs achieve a quantum advantage in a non-trivial way? Investigating the state complexity of ``non-uniform'' QFAs is another interesting direction (see \cite{VilYam15} as an example to measure the state complexity (of the restricted QFA models) by fixing the input length). 

Compared to one-way models, two-way QFA models have not been widely examined and there are many open problems related to them. 
Furthermore, promise problems, interactive proof systems, and computation with advice are new hot topics having connections with computational complexity. Further research on them will likely provide new insights. Another promising direction is connections of QFAs with algebra and using algebraic methods to study the power of QFAs.

\section*{Acknowledgements}

We are grateful to A. C. Cem Say and John Watrous for their helpful comments on the subject matter of this chapter. We would like to thank our anonymous referee for his/her helpful comments and Narad Nampersad for his suggestions to improve the language of the chapter. We also would like to thank Marats Golovkins, Paulo Mateus, Emmanuel Jeandel, Carlo Mereghetti, Farid Ablayev, Daowen Qiu, Jozef Gruska, and James P. Crutchfield for kindly answering our questions.
 
A. Yakary{\i}lmaz would like to sincerely thank his PhD. supervisor A. C. Cem Say for introducing him to the field of quantum computation and for their collaborative work where he has learned a lot and gained a great deal of experience.

A. Ambainis was supported by ERC Advanced Grant MQC and FP7 FET Proactive project QALGO. A. Yakary{\i}lmaz was partially supported by T\"{U}B\.ITAK with grant 108E142, CAPES with grant 88881.030338/2013-01, ERC Advanced Grant MQC, and FP7 FET projects QALGO and QCS.

======================================================
\bibliographystyle{abbrv}
\addcontentsline{toc}{section}{References}
\begin{footnotesize}
  \bibliography{AutomataAndQuantumComputing}
\end{footnotesize}

\newpage{\pagestyle{empty}\cleardoublepage}

\printindex

\end{document}